\RequirePackage{fix-cm}
\documentclass[twocolumn, envcountsect]{svjour3}          
\smartqed  
\usepackage{graphicx}
\usepackage[inline]{enumitem}
\usepackage[numbers, sort&compress]{natbib}  
\usepackage{multirow}
\usepackage{pifont}
\usepackage{xcolor}

\definecolor{blackkkk}{named}{black}

\usepackage{booktabs}
\usepackage{amsmath, amssymb}
\usepackage{hyperref}
\usepackage{tcolorbox}
\usepackage{bm}
\usepackage[ruled,linesnumbered,vlined]{algorithm2e}
\usepackage[english]{babel} 
\usepackage{microtype} 
\usepackage{breakurl}  
\usepackage{listings}  
\usepackage{todonotes}
\usepackage{cleveref}
\crefname{figure}{Fig.}{Figs.}
\crefname{algocf}{alg.}{algs.}
\Crefname{algocf}{Algorithm}{Algorithms}
\crefname{definition}{Definition}{Definitions}
\crefname{equation}{eq.}{eqs.}
\Crefname{equation}{Eq.}{Eqs.}

\newenvironment{proofsketch}{\paragraph{Proof sketch.}}{\hfill$\square$}
\newcommand{\yy}[1]{\textcolor{blackkkk}{#1}}

\newcommand{\MD}[1]{\textcolor{black}{#1}}
\newcommand{\algofont}[0]{\small}

\SetKwInOut{Input}{Input}
\SetKwInOut{Output}{Output}
\SetKwProg{Fn}{Function}{:}{end}
\SetKwFunction{IsDensest}{IsDensest}
\SetKwFunction{FW}{Frank-Wolfe}
\SetKwFunction{ExtSG}{ExtractSG}
\SetKwFunction{Pruning}{Pruning}
\SetKwFunction{IsLDS}{IsLDS}
\SetKwFunction{IsLTDS}{IsLTDS}
\SetKwFunction{ExtMaximal}{ExtMaximal}

\begin{document}

\title{Finding Locally Densest Subgraphs: Convex Programming with Edge and Triangle Density
}


\author{Yi Yang         \and
  Chenhao Ma \and
  Reynold Cheng\and
  Laks V.S. Lakshmanan\and
  Xiaolin Han
}


\institute{Yi Yang \at
  The Chinese University of Hong Kong, Shenzhen
  \email{yiyang3@link.cuhk.edu.cn}
  \and
  Chenhao Ma \at
  The Chinese University of Hong Kong, Shenzhen \\
  \email{machenhao@cuhk.edu.cn}           
  \emph{Corresponding author}  %
  \and
  Yixiang Fang \at
  The Chinese University of Hong Kong, Shenzhen \\
  \email{fangyixiang@cuhk.edu.cn}
  \and
  Reynold Cheng \at Department of Computer Science \&
  Guangdong–Hong Kong-Macau Joint Laboratory \&
  HKU Musketeers Foundation Institute of Data Science,
  The University of Hong Kong \\
  \email{ckcheng@cs.hku.hk}
  \and
  Laks V.S. Lakshmanan \at The University of British Columbia \\
  \email{laks@cs.ubc.ca}
  \and
  Xiaolin Han \at Northwestern Polytechnical University \\
  \email{xiaolinh@nwpu.edu.cn} \\
}

\date{Received: date / Accepted: date}

\maketitle

\begin{abstract}
  Finding the densest subgraph (DS) from a graph is a fundamental problem in graph databases. The DS obtained, which reveals closely related entities, has been found to be useful in various application domains such as e-commerce, social science, and biology. However, in a big graph that contains billions of edges, it is desirable to find more than one subgraph cluster that is not necessarily the densest, yet they reveal closely related vertices. In this paper, we study the locally densest subgraph (LDS), a recently proposed variant of DS. An LDS is a subgraph which is the densest among the ``local neighbors''. Given a graph $G$, a number of LDSs can be returned, which reflect different dense regions of $G$ and thus give more information than DS. The existing LDS solution suffers from low efficiency. We thus develop a convex-programming-based solution that enables powerful pruning.{\color{blackkkk} We also extend our algorithm to triangle-based density to solve LTDS problem. Based on current algorithms, we propose a unified framework for the LDS and LTDS problems.} Extensive experiments on {\yy thirteen} real large graph datasets show that our proposed algorithm is up to four orders of magnitude faster than the state-of-the-art.
\end{abstract}


\section{Introduction}
In modern systems and platforms that manage intricate relationship among objects, graphs have been widely used to model such relationship information~\cite{gionis2013piggybacking,dourisboure2007extraction,ching2015one,li2021analyzing,han2022leveraging,han2022deeptea,han2019traffic,han2020traffic,han2022framework,ma2019linc}. For example, the Facebook friendship network can be treated as a graph by mapping users to vertices and friendships among users to edges connecting vertices \cite{ching2015one}. \Cref{fig:graph} depicts a graph of friendship, where {\tt a} and {\tt f} have an edge meaning that they are friends of each other. In biology, graphs can be used to capture complex interactions among different proteins \cite{stelzl2005human} and relationships in genomic DNA \cite{fratkin2006motifcut}.

Lying at the core of large scale graph data mining, the densest subgraph (DS) problem \cite{goldberg1984finding,fang12efficient,qin2015locally,bahmani2012densest} is about the finding of a ``dense'' subgraph from a graph $G$ with $n$ vertices and $m$ edges. For example, the DS of \Cref{fig:graph} is the subgraph induced by vertex subset $S_{1}$, because its density, i.e., the average number of edges over the number of vertices in the subgraph, is the highest among all possible subgraphs of $G$. The DS has been found useful to various application domains~\cite{tsourakakis2021dense}. For example, dense subgraphs can be used to detect communities~\cite{chen2010dense,tsourakakis2013denser} and discover fake followers \cite{hooi2016fraudar} in social networks. In biology, the DS found can be used to identify regulatory motifs in genomic DNA \cite{fratkin2006motifcut} and find complex patterns in gene annotation graphs \cite{saha2010dense}. In graph databases, a DS is used to construct index structures for supporting reachability and distance queries \cite{jin20093}. In system optimization, DS plays an important role in social piggybacking \cite{gionis2013piggybacking,tsourakakis2021dense}, which improves the throughput of social networking systems (e.g., Twitter). 

\yy{
%
Beyond the edge-based setting, several alternative notions of density have been explored in the literature~\cite{tsourakakis2013denser, tsourakakis2015k, mitzenmacher2015scalable}. Among them, {\em triangle-based density} has gained particular attention in recent years. Unlike edges, triangles represent stronger and more cohesive connections between nodes, making them especially relevant for capturing community structures and patterns in graphs. Triangles frequently emerge across various domains. For instance, in social networks, a triangle models a tightly-knit group, as friends of friends are often friends themselves. Consequently, identifying subgraphs with high triangle density can reveal richer structures and relationships that edge-based measures may overlook. This characteristic makes triangle-dense subgraphs valuable for various data mining tasks, offering deeper insights by considering higher-order connectivity.}

\begin{figure}[tb]
    \includegraphics[width=0.43\textwidth]{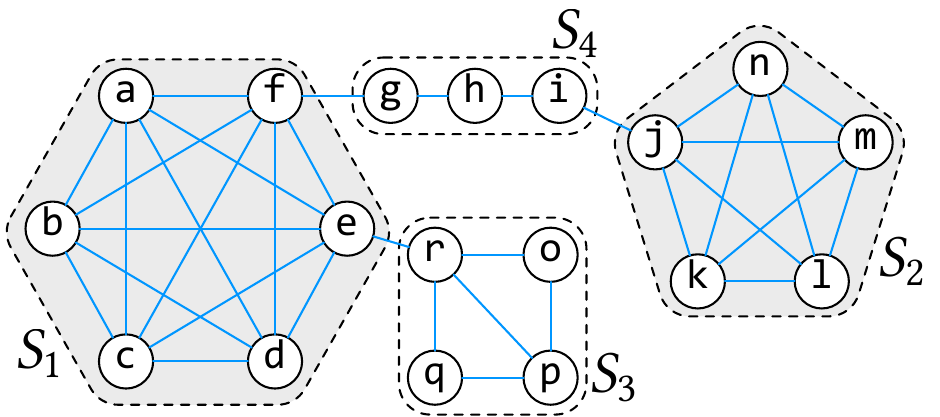}
    \caption{An undirected graph, $G$}
    \label{fig:graph}
\end{figure}

{\bf LDS model.} Most existing studies \cite{goldberg1984finding,bahmani2012densest,fang12efficient,charikar2000greedy} on the DS problem focus on finding the densest subgraph. In fact, it is not uncommon to find more than one ``dense subgraphs''. For example, in community detection, it is interesting to explore multiple communities in a social network, even if not all communities have the highest density. Motivated by this, \citeauthor{qin2015locally} \cite{qin2015locally} proposed the locally densest subgraph (LDS) model \cite{qin2015locally,samusevich2016local,tsourakakis2021dense}, which identifies LDSs of the graph. An LDS needs to be {\it dense} and {\it compact}. Conceptually, an LDS is a subgraph with the highest density in its vicinity. Moreover, LDS is ``compact'', in the sense that any subset of its vertices is highly connected to each other. (We will discuss the formal definition of LDS in \Cref{sec:prelim}.) For \Cref{fig:graph}, the two subgraphs induced by vertex subsets $S_{1}$ and $S_{2}$, denoted as $G[S_{1}]$ and $G[S_{2}]$ respectively, are two LDSs of $G$. On the other hand, the subgraph $G[S_{1}\cup S_{3}]$ is not an LDS.






The LDS problem has three important properties~\cite{qin2015locally}:

\noindent $\bullet$ The LDS is {\it parameter-free}. As we will explain later, the two LDSs $G[S_{1}]$ and $G[S_{2}]$ can be obtained from $G$ in \Cref{fig:graph} without setting density threshold or other parameters.

\noindent $\bullet$ For any pair of LDSs on a given graph, they do not have common vertices (i.e., {\it disjoint}). In the above example, the two LDSs $G[S_{1}]$ and $G[S_{2}]$ are disjoint. This allows us to identify all the non-overlapping ``dense regions'' of a graph.

\noindent $\bullet$ The set of subgraphs found by the LDS problem is a superset of the subgraph found by the DS problem. For example, DS only identifies $G[S_{1}]$ in \Cref{fig:graph}, with a density of $2.5$. However, LDS returns $\{G[S_{1}], G[S_{2}]\}$. Notice that $G[S_{2}]$, which has a density of $2$, does not overlap with $G[S_{1}]$.

\begin{figure}[t]
    \includegraphics[width=0.48\textwidth]{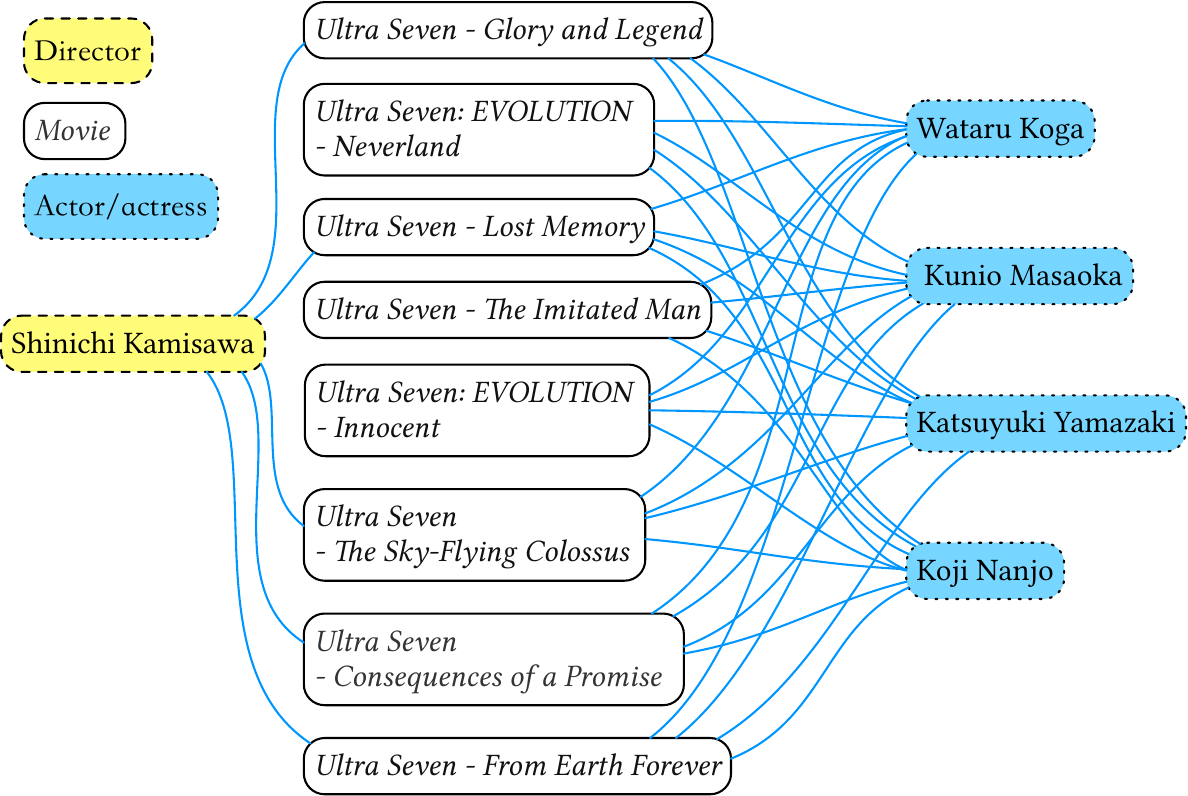}
    \caption{An LDS about ``Ultraman''.}
    \label{fig:case-study}
\end{figure}

We have conducted a case study on a large movie graph provided by TCL, an electronic product provider in China. We found that the LDSs found are about different topics (e.g., western films, Chinese martial fiction, Danish comedies). \Cref{fig:case-study} shows an LDS, with the third highest density, about the Japanese sci-fi series ``Ultraman''. The DS model can only find a subgraph about western films. The case study shows that the LDS model can find representative dense subgraphs of a graph.

    {\bf Prior work.} Based on the LDS model, \citeauthor{qin2015locally} proposed a max-flow-based LDS solution with pruning techniques built on $k$-core, named {\tt LDSflow} \cite{qin2015locally}. Here $k$-core \cite{seidman1983network} is a cohesive subgraph, which requires each vertex to have at least $k$ neighbors in the subgraph. To find LDSs, {\tt LDSflow} adopts a prune-and-verify framework. Specifically, the pruning bounds for vertices based on $k$-core are derived, which are then used to prune vertices. The LDSs are then obtained by verifying the remaining vertices through max-flow computation on the flow network. \yy{Inspired by LDS model, \citeauthor{samusevich2016local} \cite{samusevich2016local} generalized it to a triangle-based setting, which is referred to as the locally triangle-densest subgraphs (LTDS) model. A max-flow-based algorithm, named {\tt LTDSflow}, was also proposed to identify the top-$k$ LTDSs with the largest triangle-based density.} The main problem of {\tt LDSflow} and {\tt LTDSflow} is that they may not scale well on large graphs. For example, {\tt LDSflow} needs around 17 hours to output $15$ LDSs with the highest densities from a graph with around 3 million edges.

    {\bf Contributions.} We have developed an efficient and scalable LDS solution, as detailed below:
\begin{enumerate}[wide, labelindent=0pt, topsep=2pt, itemsep=2pt]
    \item {\em Propose the compact number.} Given a vertex $v$, we define the {\it compact number} of $v$, which depicts the degree of compactness of the most compact subgraph containing $v$.
          We further use compact numbers to link the LDS problem and a convex program for the DS problem theoretically, leveraging the fact that vertices in an LDS share the same compact number, and compact numbers can be computed by solving the convex program.
    \item {\em An efficient convex-programming-based algorithm with elegant pruning techniques for solving the LDS problem.} We propose a prune-and-verify algorithm based on the convex programming, which is named {\tt LDScvx}. For pruning, we show that an iterative Frank-Wolfe algorithm \cite{frank1956algorithm} can provide the tight upper and lower bounds for compact numbers, which can be used to prune vertices not contained by any LDS from the graph. For verification, {\tt LDSflow} \cite{qin2015locally} performs min-cut computation based on a specific $k$-core of $G$. In {\tt LDScvx}, we only need to compute the min-cut based on a smaller subgraph of the $k$-core by exploiting the upper and lower bounds of compact numbers.
    \yy{\item {\em An extended convex-programming-based algorithm that generalizes the LDS solution to the LTDS problem.} We adapted our {\tt LDScvx} solution for the LTDS problem. This extension preserves the efficiency and pruning strategies of {\tt LDScvx} while ensuring adaptability to the structural differences between LDS and LTDS.
    \item {\em A unified framework for LDS problem.} Inspired by \cite{zhou2024depth}, we summarize a unified framework for all LDS and LTDS solutions from a high-level perspective. Specifically, this framework consists of four stages: {\it Initial Reduction}, {\it Vertex Weight Updating}, {\it Graph Reduction and Division}, and {\it Candidate Subgraph Extraction and Verification}. Our framework allows us to compare our methods with existing solutions comprehensively.}
    \item {\em Extensive experiments.} We have experimentally compared our CP-based algorithm with the existing algorithm on thirteen real large datasets with sizes of up to 1.6 billion edges. Our results show that {\tt LDScvx} is up to four orders of magnitude faster than the state-of-the-art. A case study on a large movie graph has also been performed.
\end{enumerate}

{\bf Outline.} The rest of the paper is organized as follows. We review the related work in \Cref{sec:related}. In \Cref{sec:prelim}, we formally present the LDS problem. \Cref{sec:cn} defines the compact number and discusses its relationship with LDS and convex programming. We present our LDS algorithm in \Cref{sec:algo} and extension for LTDS problem in \Cref{sec:ltds}. The unified framework and empirical results are shown in \Cref{sec:framework} and \Cref{sec:exp}, respectively. \Cref{sec:conc} concludes the paper.

\section{Related Work}
\label{sec:related}
The densest subgraph is regarded as one kind of cohesive subgraph. \citeauthor{tsourakakis2021dense} \cite{tsourakakis2021dense} provides a comprehensive study of techniques and applications of the DS problem. Other related topics include $k$-core \cite{seidman1983network}, $k$-truss \cite{wang2012truss}, clique, and quasi-clique \cite{conte2018d2k}. More details on them can be found in \cite{fang2020survey,fang2022densest}. In the following, we focus on densest subgraph discovery and its variants.

    {\bf Densest subgraph discovery (DS).} \citeauthor{goldberg1984finding} \cite{goldberg1984finding} introduced the densest subgraph (DS) problem on undirected graphs, which aims to find the subgraph whose edge-density is the highest among all the subgraphs where the edge-density of a graph $G=(V, E)$ is defined as $\frac{|E|}{|V|}$. To solve the DS problem, \citeauthor{goldberg1984finding} \cite{goldberg1984finding} proposed an exact algorithm based on flow network via solving $O(\log(n))$ min-cut problems.
%
%
Later, \citeauthor{fang12efficient} \cite{fang12efficient} improved the efficiency of the flow-based exact algorithm by locating the DS in a specific $k$-core.
Exact algorithms can process small graphs reasonably, but they cannot scale well to large graphs. To remedy this issue, several approximation algorithms \cite{charikar2000greedy,bahmani2012densest,fang12efficient,boob2020flowless,zhou2024counting} have been developed.
%

\citeauthor{kannan1999analyzing} \cite{kannan1999analyzing} extended the DS problem definition to directed graphs.
\citeauthor{charikar2000greedy} \cite{charikar2000greedy} developed an exact polynomial-time algorithm to find the directed densest subgraph via $O(n^2)$ linear programs.
\citeauthor{khuller2009finding} \cite{khuller2009finding} presented a max-flow-based exact algorithm.
\citeauthor{ma2020efficient} \cite{ma2020efficient,ma2021directed,ma2021efficient} improved the max-flow-based algorithm by introducing the notion of $[x,y]$-core and exploiting a divide-and-conquer strategy.
To further boost the efficiency for large-scale graphs, approximation algorithms  \cite{charikar2000greedy,khuller2009finding,ma2020efficient,sawlani2020near,ma2022convex} for the directed DS problem are also developed.
\yy{Based on existing DS algorithms for both undirected and directed graphs, \citeauthor{zhou2024depth} \cite{zhou2024depth} proposed a unified framework that integrates all algorithms from a high-level perspective.}
%

Further, there are some variants of the DS problem focusing on different aspects.
\citeauthor{asahiro2002complexity} \cite{asahiro2002complexity} studied the DS problem with constraints on the size of the DS. However, the size constraints (e.g., at least $k$ vertices or at most $k$ vertices to be included) make the DS problem NP-hard \cite{asahiro2002complexity}. Hence, \citeauthor{andersen2009finding} \cite{andersen2009finding}, \citeauthor{khuller2009finding} \cite{khuller2009finding}, and \citeauthor{chekuri2022densest} \cite{chekuri2022densest} studied efficient approximation algorithms on the variants of the size-constrained DS.
\citeauthor{tsourakakis2015k} extended the notion of density based on edges to $k$-clique-density, and studied the DS problem with $k$-clique-density \cite{tsourakakis2015k,mitzenmacher2015scalable,sun2020kclist++}.
\citeauthor{chang2020deconstruct} \cite{chang2020deconstruct} proposed a novel and efficient index to report all minimal DS's and enumerate all DS's, where the minimal DS is strictly denser than all of its proper subgraphs.
%
Tatti et al. \cite{tatti2015density} and Danisch et al. \cite{danisch2017large} studied the density-friendly decomposition problem, which decomposes a graph into a chain of subgraphs, where each subgraph is nested within the next one, and the inner one is denser than the outer ones.
\citeauthor{galbrun2016top} \cite{galbrun2016top} and \citeauthor{dondi2021top} \cite{dondi2021top,dondi2021novel} studied the densest subgraphs with overlaps.

    {\bf Locally densest subgraph (LDS).} \citeauthor{qin2015locally} \cite{qin2015locally} proposed a new DS model, named locally densest subgraph (LDS). Based on this model, users can identify all the locally densest regions of a graph. \citeauthor{qin2015locally} have shown that such subgraphs cannot be found by other DS models theoretically and empirically. Compared to the original DS model, \citeauthor{qin2015locally} \cite{qin2015locally} showed that the subgraphs provided by the LDS model best represent different local dense regions of the graph, while the DS model only provides the DS.
Furthermore, they compared the LDS model with a straightforward greedy approach based on the DS, which finds the DS \cite{goldberg1984finding} at a time, removes it from the graph, and repeats the process for $k$ times. Nevertheless, this greedy approach has several shortcomings \cite{qin2015locally}:
\begin{enumerate*}
    \item The top-$k$ results may not fully reflect the top-$k$ densest regions. Especially when the graph has a vast dense region, subgraphs in other dense regions may have a low chance to appear.
    \item A subgraph returned by the greedy approach can be partial and contained by a better subgraph.
    \item This greedy approach is essentially a heuristic and does not allow for a formal characterization of the result.
\end{enumerate*}
The above claims are also verified by the experimental results in \cite{qin2015locally}. With the LDS model justified, \citeauthor{qin2015locally} \cite{qin2015locally} proposed a max-flow-based algorithm, {\tt LDSflow}, to find the top-$k$ LDSs from a graph.
However, we found that {\tt LDSflow} \cite{qin2015locally} does not scale to large graphs \MD{because {\tt LDSflow} needs to run the max-flow algorithm on the graph several times to find an LDS candidate and verify it, and the max-flow computation can be quite time-consuming. For example, the state-of-the-art max-flow algorithm is proposed by \citeauthor{orlin2013max} with time complexity of $O(nm)$ \cite{orlin2013max}.}
In this paper, we propose an efficient and scalable LDS algorithm for finding the top-$k$ LDSs from  large graphs. We note that   \citeauthor{samusevich2016local} \cite{samusevich2016local} extended the LDS model from edge-based density to triangle-based density. We focus on  edge-based density, and leave triangle density based LDS for future work.


\section{Preliminaries}
\label{sec:prelim}

This section formally defines the locally densest subgraph problem (LDS problem). \Cref{tab:notations} lists the notations used in this paper.

\begin{table}[tb]
    \centering
    \algofont
    \caption{Notations and meanings.}
    \begin{tabular}{c|p{5.8cm}}
        \hline
        Notation                & Meaning                                            \\
        \hline
        \hline
        $G=(V,E)$               & a graph with vertex set $V$ and edge set $E$       \\
        \hline
        $G[S]$                  & the subgraph induced by $S$                        \\
        \hline
        $d_{G}(u)$              & the degree of a vertex $u$ in $G$                  \\
        \hline
        $\mathsf{density}(G)$   & the density of graph $G$, i.e., $\frac{|E|}{|V|}$  \\
        \hline
        $\phi(u), \phi_{G'}(u)$ & compact number of $u$ in $G$ and $G'$ respectively \\
        \hline
        $\overline{\phi}(u)$    & the upper bound of $\phi(u)$                       \\
        \hline
        $\underline{\phi}(u)$   & the lower bound of $\phi(u)$                       \\
        \hline
        $\mathsf{core}_{G}(u)$  & the core number of $u$ in $G$                      \\
        \hline
        $\mathsf{CP}(G)$        & the convex program of $G$ for densest subgraph     \\
        \hline
        $\bm{\alpha}$           & the weights distributed from edges to vertices     \\
        \hline
        $\bm{r}$                & the weights received by each vertex                \\
        \hline
    \end{tabular}%
    \label{tab:notations}%
\end{table}

Let $G=(V, E)$ be an undirected graph with $n=|V|$ vertices and $m=|E|$ edges. For each vertex $u\in G$, we use $d_{G}(u)$ to represent its degree in $G$. Given a subset $S\subseteq V$, $E(S)$ denotes the set of edges induced by $S$, i.e., $E(S)=E\cap (S\times S)$. Hence, the subgraph induced by $S$ is denoted by $G[S]=(S, E(S)).$ Following the classic graph density definition \cite{goldberg1984finding,charikar2000greedy,asahiro2000greedily,bahmani2012densest,qin2015locally}, the density of a graph $G=(V,E)$, denoted by $\mathsf{density}(G)$, is defined as:
\begin{equation}
    \label{equ:density}
    \mathsf{density}(G)= \frac{|E|}{|V|}.
\end{equation}
Based on the density definition, the densest subgraph problem is to find the subgraph $G[S]$ of $G$ such that $\mathsf{density}(G[S])$ is maximized.

Densest subgraph discovery is widely applied in many graph mining tasks (e.g., \cite{fratkin2006motifcut,saha2010dense,hooi2016fraudar,ma2020efficient,ma2021directed}). However, as pointed out by \cite{qin2015locally}, it is usually not sufficient to find one dense subgraph, in many applications such as community detection \cite{dourisboure2007extraction}. \citeauthor{qin2015locally} \cite{qin2015locally} proposed a locally densest subgraph model to provide multiple dense subgraphs, which are dense and compact. \Cref{equ:density} gives the density of a subgraph. For compactness, \Cref{def:rho-compact} defines what is a $\rho$-compact subgraph, which comes from the intuition that a graph is compact if any subset of vertices is highly connected to others in the graph \cite{qin2015locally}.

\begin{definition}[$\rho$-compact \cite{qin2015locally}]
    \label{def:rho-compact}
    A graph $G=(V, E)$ is $\rho$-compact if and only if $G$ is connected, and removing any subset of vertices $S \subseteq V$ will result in the removal of at least $\rho \times |S|$ edges in $G$, where $\rho$ is a non-negative real number.
\end{definition}

\begin{definition}[Maximal $\rho$-compact subgraph \cite{qin2015locally}]
    \label{def:maximal-rho}
    A $\rho$-compact subgraph $G[S]$ of $G$ is a maximal $\rho$-compact subgraph of $G$ if and only if there does not exist a supergraph $G[S']$ of $G[S]$ ($S'\neq S$) in $G$ such that $G[S']$ is also $\rho$-compact.
\end{definition}

\noindent\MD{{\bf Remark.} The locally dense subgraph in \cite{tatti2015density} is similar to the maximal $\rho$-compact subgraph. But they are different because the $\rho$-compact subgraph needs to be connected.}

\begin{definition}[Locally densest subgraph \cite{qin2015locally}]
    \label{def:lds}
    A subgraph $G[S]$ of $G$ is a locally densest subgraph (LDS) of $G$ if and only if $G[S]$ is a maximal $\mathsf{density}(G[S])$-compact subgraph in $G$.
\end{definition}

From \Cref{def:lds}, we can find that
\begin{enumerate*}
    \item the definition itself is parameter-free;
    \item an LDS is compact;
    \item any supergraph of an LDS cannot be more compact than the LDS itself;
    \item any subgraph $G[S']$ of an LDS $G[S]$ cannot be denser than the LDS itself \cite{qin2015locally}.
\end{enumerate*}
These properties show that an LDS is indeed a {\em locally densest} subgraph.
The third one comes from that an LDS $G[S]$ is a maximal $\rho$-compact subgraph.
The last one can be proven by contradiction. Suppose that $G[S']$ has a density satisfying $\mathsf{density}(G[S'])>\mathsf{density}(G[S])$. If we remove vertex set $U=S\setminus S'$ from $G[S]$, the number of edges removed is $\mathsf{density}(G[S])\times |S| - \mathsf{density}(G[S'])\times |S'|< \mathsf{density}(G[S])\times (|S|-|S'|)=\mathsf{density}(G[S])\times |U|$, which contradicts that $G[S]$ is $\mathsf{density}(G[S])$-compact.
We further illustrate the LDS definition with the following example.

\begin{example}[LDS]
    Consider the graph $G$ shown in \Cref{fig:graph}. The subgraph $G[S_{1}]$ with density $\frac{5}{2}$ is a maximal $\frac{5}{2}$-compact subgraph. Hence, $G[S_{1}]$ is an LDS. Similarly, $G[S_{2}]$ with density $2$ is also an LDS as it is a maximal $2$-compact subgraph.
    The subgraph $G[S_{3}]$ with density $\frac{5}{4}$ is a $\frac{5}{4}$-compact subgraph. But $G[S_{3}]$ is not an LDS because it is contained in $G[S_{1}\cup S_{3}]$ which is $\frac{3}{2}$-compact (and also $\frac{5}{4}$-compact).
    $G[S_{1}\cup S_{3}]$ is also not an LDS, because its density is $\frac{21}{10}$ but it is not a $\frac{21}{10}$-compact subgraph. The compactness of $G[S_{1}\cup S_{3}]$ is $\frac{3}{2}=\frac{6}{4}$, because removing $S_{3}$ from $G[S_{1}\cup S_{3}]$ will result in removing of 6 edges.
    \qed
\end{example}

The LDS has a useful property that all the LDSs in a graph $G$ are pairwise disjoint.
\begin{lemma}[Disjoint property \cite{qin2015locally}]
    \label{lem:disjoint}
    Suppose $G[S]$ and $G[S']$ are two LDSs in $G$, we have $S\cap S' = \emptyset$.
\end{lemma}

According to \Cref{lem:disjoint}, the LDS model can be used to find all dense regions of a graph. However, most real-world applications (e.g., community detection) usually require finding the top-$k$ dense regions of a graph \cite{qin2015locally}. Hence, we focus on finding the top-$k$ LDSs with the largest densities following \cite{qin2015locally}.

\begin{problem}[LDS problem \cite{qin2015locally}]
Given a graph $G$ and an integer $k$, the LDS problem is to compute the top-$k$ LDSs with the largest density in $G$.
\end{problem}


\section{From Compactness to CP}
\label{sec:cn}
In this section, we first propose a new concept named {\em compact number}, inspired by the core number related to $k$-core \cite{seidman1983network}, which is a cohesive subgraph model. Then, we present some interesting properties of LDS from the perspective of compact numbers. Afterward, we show that the compact numbers can be computed via a convex programming (CP) formulation of the densest subgraph problem. Hence, the compact number acts as a bridge between the LDS problem and the CP formulation, as shown in \Cref{fig:cn-bridge}.

\begin{figure}[ht]
    \includegraphics[width=0.48\textwidth]{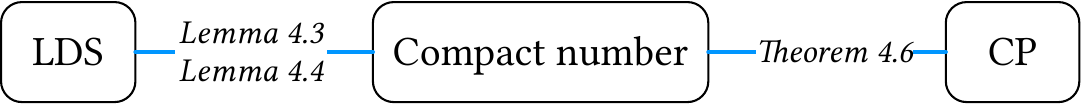}
    \caption{Relation among LDS, Compact number, and CP.}
    \label{fig:cn-bridge}
\end{figure}


\subsection{Compact Number and LDS}
\label{sec:cn:cn-lds}

Inspired by the core number of $k$-core \cite{seidman1983network}, we define the {\em compact number} of each vertex in graph $G$ with respect to $\rho$-compact subgraphs of $G$.

\begin{definition}[Compact number]\label{def:cn}
    Given a graph $G=(V,E)$, the compact number of each vertex $u \in V$, denoted by $\phi(u)$, is the largest $\rho$ such that $u$ is contained in a $\rho$-compact subgraph of $G$.
\end{definition}

We use $\overline{\phi}(u)$ and $\underline{\phi}(u)$ to denote the upper and lower bounds of $\phi(u)$ in $G$, respectively. Besides, $\overline{\phi}_{G'}(u)$ denotes the upper bound of $\phi_{G'}(u)$ in $G'$, $G'$ is a subgraph of $G$.

\begin{example}[Compact number]
    Consider vertex {\tt q} of $G$ in \Cref{fig:graph}. The compact number of {\tt q} is $\frac{3}{2}$, i.e., $\phi(${\tt q}$)=\frac{3}{2}$, because $G[S_{1} \cup S_{3}]$ is a $\frac{3}{2}$-compact subgraph containing {\tt q} and there is no other subgraph containing {\tt q} with a larger $\rho$. Removing $S_{3}$ from $G[S_{1} \cup S_{3}]$ will remove 6 edges, so $G[S_{1} \cup S_{3}]$ is $\frac{3}{2}$-compact.
\end{example}

Next, we discuss the relationship between the LDS and the compact numbers of vertices within or adjacent to the LDS via the following lemmas.

\begin{lemma}
    \label{lem:lds-in-phi}
    Given an LDS $G[S]$ in $G$, $\forall u \in S$, we have $\phi(u)=\mathsf{density}(G[S])$.
\end{lemma}
\begin{proof}
    As $G[S]$ is a maximal $\mathsf{density}(G[S])$-compact subgraph, for each $u\in S$, there exists no other subgraph $G[S']$ containing $u$ such that $G[S']$ is a $\rho$-compact subgraph with $\rho>\mathsf{density}(G[S])$. We prove the claim by contradiction. Suppose $G[S']$ is a $\rho$-compact subgraph with $\rho > \mathsf{density}(G[S])$ and $u\in S'$, we have $\mathsf{density}(S') \geq \rho > \mathsf{density}(G[S])$. First $S' \subseteq S$, because $G[S]$ is a maximal $\mathsf{density}(G[S])$-compact subgraph and $S'\cap S \neq \emptyset$. If we remove $U=S\setminus S'$ from $G[S]$, the number of edges removed is $|E(S)|-|E(S')|=\mathsf{density}(G[S])\times |S| - \mathsf{density}(G[S'])\times |S'| < \mathsf{density}(G[S])\times(|S|-|S'|)=\mathsf{density}(G[S])\times |U|$. This contradicts that $G[S]$ is $\mathsf{density}(G[S])$-compact. Hence, $\mathsf{density}(G[S])$ is the compact number of all vertices in $S$.
\end{proof}

\begin{lemma}
    \label{lem:lds-out-phi}
    Given an LDS $G[S]$ in $G$, $\forall (u,v) \in E$, if $u\in S$ and $v\in V\setminus S$, we have $\phi(u)>\phi(v)$.
\end{lemma}

\begin{proof}
    The lemma follows from the LDS definition. Suppose $\exists (u,v) \in E, u\in S, v\in V\setminus S$, $\phi(u) \leq \phi(v)$, which contradicts that $G[S]$ is a maximal $\mathsf{density}(G[S])$-compact subgraph.
\end{proof}

\Cref{lem:lds-in-phi,lem:lds-out-phi} indicate that the compact numbers of all vertices in an LDS $G[S]$ are exactly $\mathsf{density}(G[S])$ and the compact numbers of all vertices adjacent to vertices in $G[S]$ but not in $G[S]$ are less than $\mathsf{density}(G[S])$, respectively. We further illustrate the two lemmas via the following example.

\begin{example}
    Consider the LDS $G[S_{1}]$ of $G$ in \Cref{fig:graph}. We can see that $\forall u \in S_{1}$, $\phi(u) = \mathsf{density}(G[S_{1}]) = \frac{5}{2}$, which fulfills \Cref{lem:lds-in-phi}. For vertices locating outside $G[S_{1}]$ but adjacent to some vertices in $S_{1}$, i.e., {\tt g}, {\tt r}, and {\tt q}, their compact numbers satisfies \Cref{lem:lds-out-phi}: $\phi(${\tt g}$)=\frac{4}{3} < \frac{5}{2} $, $\phi(${\tt r}$)=\phi(${\tt q}$)=2 < \frac{5}{2}$.
\end{example}

According to \Cref{lem:lds-in-phi,lem:lds-out-phi}, compact numbers are powerful to extract and verify LDSs from a graph $G$. If the compact numbers of all vertices are ready, we can partition the graphs into different subgraphs, where the vertices within the same subgraph share the same compact number and are connected, based on \Cref{lem:lds-in-phi}. Then, \Cref{lem:lds-out-phi} can be used to select LDSs from all subgraphs. Hence, the efficient computation of compact numbers is a key issue. To tackle this issue, we show that compact numbers can be obtained by solving a convex program in the next subsection.

\subsection{Compact Number and CP}
\label{sec:cn:cn-cp}

In this subsection, we first review the convex program (CP) for densest subgraphs by \citeauthor{danisch2017large}~\cite{danisch2017large}. Next, we theoretically prove that compact numbers can be obtained by solving the convex program.

\begin{equation}\label{equ:cp}
    \begin{aligned}
        \mathsf{CP}(G) &  & \min                        & \sum_{u \in V} r_{u}^{2}          &                     \\
                       &  & r_{u}                       & = \sum_{(u,v)\in E} \alpha_{u,v}, & \forall u \in V     \\
                       &  & \alpha_{u,v} + \alpha_{v,u} & \geq 1,                           & \forall (u,v) \in E \\
                       &  & \alpha_{u,v}, \alpha_{v,u}  & \geq 0.                           & \forall (u,v) \in E
    \end{aligned}
\end{equation}
The intuition of \cref{equ:cp} is that each edge $(u,v) \in E$ tries to distribute its weight, i.e., $1$, between its two endpoints $u$ and $v$ such that the weight sum received by the vertices are as even as possible. Because in the DS $G[S]$ of $G$, it is possible to distribute all edge weights such that the weight sum received by each vertex in $S$ is exactly $\mathsf{density}(G[S])=\frac{|E(S)|}{|S|}$. Following the intuition, $\alpha_{u,v}$ in \cref{equ:cp} indicates the weight assigned to $u$ from edge $(u,v)$, and $r_{u}$ is the weight sum received by $u$ from its adjacent edges. \MD{\citeauthor{danisch2017large} \cite{danisch2017large} used $r$ and $\alpha$ of \cref{equ:cp} to tentatively decompose the graph into a chain of subgraphs and applied max-flow to fine-grain and confirm the partitions such that each subgraph is nested within the next one with densities in descending order.}

Before proving the compact numbers can be obtained via solving \cref{equ:cp}, we briefly review the Frank-Wolfe-based iterative algorithm proposed by \citeauthor{danisch2017large}~\cite{danisch2017large} for optimizing \cref{equ:cp}. \FW (\Cref{alg:fw}) outlines the steps to optimize \cref{equ:cp}. \FW first initializes $\bm{\alpha}$ and $\bm{r}$ \MD{(lines 2--3)}. Then, in each iteration, each edge $(u,v)\in E$ attempts to distribute its weight, i.e., $1$, to the endpoint with a smaller $r$ value (lines 4--10).

\begin{algorithm}[t]
    \caption{Frank-Wolfe-based algorithm \cite{danisch2017large}}
    \label{alg:fw}
    \algofont
    \Fn{\FW{$G=(V,E)$, $N\in \mathbb{Z}_{+}$}}{
        \lForEach{$(u,v) \in E$}{$\alpha^{(0)}_{u,v} \gets \frac{1}{2}$, $\alpha^{(0)}_{v,u} \gets \frac{1}{2}$}
        \lForEach{$u \in V$}{$r^{(0)}_{u} \gets \sum_{(u,v)\in E} \alpha^{(0)}_{u,v}$}
        \For{$i=1, \dots, N$}{
            $\gamma_{i} = \frac{2}{i+2}$\;
            \ForEach{$(u,v)\in E$}{
                \lIf{$r^{(i-1)}_{u} < r^{(i-1)}_{v}$}{$\hat{\alpha}_{u,v} \gets 1$, $\hat{\alpha}_{v,u} \gets 0$}
                \lElse{$\hat{\alpha}_{u,v} \gets 0$, $\hat{\alpha}_{v,u} \gets 1$}
            }
            $\bm{\alpha}^{(i)} \gets (1-\gamma_{i}) \cdot \bm{\alpha}^{(i-1)} + \gamma_{i} * \hat{\bm{\alpha}}$\;
            \lForEach{$u \in V$}{$r^{(i)}_{u} \gets \sum_{(u,v)\in E} \alpha^{(i)}_{u,v}$}
        }
        \Return{$(\bm{r}^{(i)}, \bm{\alpha}^{(i)})$}\;
    }
\end{algorithm}

Next, we prove that the compact numbers can be extracted from the optimal solution of \cref{equ:cp}.

\begin{theorem}\label{thm:phi-r}
    Suppose $(\bm{r}^{*}, \bm{\alpha}^{*})$ is an optimal solution of \cref{equ:cp}. Then, each $r_{u}^{*}$ in $\bm{r}^{*}$ is exactly the compact number of $u$, i.e., $\forall u \in V, r_{u}^{*}=\phi(u)$.
\end{theorem}

\begin{proof}
    For a vertex $u\in V$, let $X=\{v \in V|r_{v}^{*} > r_{u}^{*}\}$, $Y=\{v \in V|r_{v}^{*} = r_{u}^{*}\}$, and $Z=\{v\in V|r_{v}^{*} < r_{u}^{*}\}$. Clearly, $u \in Y$.

    We first prove $G[X\cup Y]$ is a $r_{u}^{*}$-compact subgraph. Removing $Y$ from $G[X \cup Y]$ will result in the removal of $r_{u}^{*}\times |Y|$ edges in $G[X \cup Y]$. The optimality of $\bm{r}^{*}$ implies that
    \begin{enumerate}
        \item $\forall (x,y) \in E\cap(X\times Y)$, $r_{x}^{*} > r_{y}^{*}$ and $\alpha_{x,y}=0$;
        \item $\forall (y,z) \in E\cap(Y\times Z)$, $r_{y}^{*} > r_{z}^{*}$ and $\alpha_{y,z}=0$.
    \end{enumerate}
    \MD{Otherwise, suppose $\exists (x,y) \in E\cap(X\times Y)$ such that $\alpha_{x,y}>0$ without loss of generality. There exists $r_{x}^{*} - r_{y}^{*}>\epsilon>0$ such that we can reduce $\alpha_{x,y}$ and increase $\alpha_{y,x}$ by $\epsilon$, respectively. Hence, $r_{x}^{*}$ is reduced and $r_{y}^{*}$ is increased by $\epsilon$, respectively. After such modification, the objective function be decreased by $2\epsilon(r_{x}^{*}-r_{y}^{*}-\epsilon)$, which contradicts the optimality of $r^{*}$.}
    Hence, \MD{$r_{u}^{*} \times |Y|=\sum_{(y,x)\in E \wedge y \in Y}{\alpha_{y,x}}=|((X\times Y) \cup (Y\times Y)) \cap E|$,} which is exactly the number of edges to be removed when removing $Y$ from $G[X\cup Y]$.
    Besides, removing any $Q \subseteq X\cup Y$ from $G[X\cup Y]$ will result in the removal of at least $r_{u}^{*}\times |Q|$ edges, \MD{because $\sum_{(s,t)\in E(X\cup Y) \wedge s \in Q} 1 \geq \sum_{(s,t)\in E \wedge s \in Q}\alpha_{s,t} \geq r_{u}^{*}\times |Q|$, where the second inequality follows from the first condition of \cref{equ:cp}.}
    Hence, $G[X\cup Y]$ is a $r_{u}^{*}$-compact subgraph.

    For any other subgraph $G[S]$ containing $u$, $G[S]$ is a $\phi$-compact subgraph, where $\phi \leq r_{u}^{*}$. Clearly, $S\cap (Y\cup Z) \neq \emptyset$. Removing $S\cap (Y\cup Z)$ from $G[S]$ will result in the removal of no more than $r_{u}^{*} \times |S\cap (Y\cup Z)|$ edges, which can be proved by contradiction analogously.
\end{proof}

\Cref{thm:phi-r} shows that given an optimal solution $(\bm{r}^{*}, \bm{\alpha}^{*})$ to \cref{equ:cp}, the weight received by each vertex $\bm{r}^{*}_{u}$ is exactly the compact number of $u$. \Cref{eg:phi-r} further depicts \Cref{thm:phi-r} concretely.

\begin{example}
    \label{eg:phi-r}
    Consider the convex program \cref{equ:cp} for $G$ in \Cref{fig:graph}. We list the optimal solution $(\bm{r}^{*}, \bm{\alpha}^{*})$ values in \Cref{tab:r-alpha-values}. Some values of $\bm{\alpha}^{*}$ are omitted, as they can be inferred from other $r^{*}_{u}$ and $\alpha^{*}_{u,v}$ values. For each $u\in V$, $r^{*}_{u}$ is exactly the compact number of $u$, $\phi(u)$.
\end{example}

\begin{table}[t]
    \centering
    \renewcommand{\arraystretch}{1.5}
    \caption{$(\bm{r}^{*}, \bm{\alpha}^{*})$ to \cref{equ:cp} for $G$ in \Cref{fig:graph}.}
    \algofont
    \begin{tabular}{c|c||c|c}
        \hline
        Vertices                  & $r^{*}_{u}$   & Edges                                                      & $\alpha^{*}_{u,v}$ \\
        \hline
        \hline
        $u\in S_{1}$              & $\frac{5}{2}$ & $(u,v)\in E(S_{1})\cup E(S_{2})$                           & $\frac{1}{2}$      \\
        \hline
        $u\in S_{2}$              & $2$           & ({\tt g}, {\tt f}), ({\tt i}, {\tt j}), ({\tt r}, {\tt e}) & $1$                \\
        \hline
        $u\in S_{3}$              & $\frac{3}{2}$ & ({\tt g}, {\tt h}), ({\tt i}, {\tt h})                     & $\frac{1}{3}$      \\
        \hline
        {\tt g}, {\tt h}, {\tt i} & $\frac{4}{3}$ & \multicolumn{2}{c}{$\cdots$}                                                    \\
        \hline
    \end{tabular}%
    \label{tab:r-alpha-values}%
\end{table}%

According to Corollary 4.9 of \cite{danisch2017large}, it may need many iterations for \FW to obtain the optimal $(\bm{r}^{*}, \bm{\alpha}^{*})$. Fortunately, we found that the approximate solution provided by \FW already helps prune the vertices not contained in LDSs and extract LDSs, which will be discussed in the next section.


\section{Our LDS Algorithm}
\label{sec:algo}

In this section, we introduce our convex programming based LDS algorithm, named {\tt LDScvx}. \Cref{fig:algo} presents the workflow of {\tt LDScvx}. First, we compute an approximate solution $(\bm{r}, \bm{\alpha})$ via \FW; next, we extract stable groups, which can be used to bound the compact numbers, from $G$ based on $(\bm{r}, \bm{\alpha})$ via {\tt ExtractSG}; afterward, we prune invalid vertices according to their compact numbers and generate LDS candidates via {\tt Pruning}; finally, we verify the LDS candidates via {\tt IsLDS}. If the verification failed, we repeat the above process to provide higher-quality $(\bm{r}, \bm{\alpha})$ and compact number estimation.

\begin{figure}[htb]
    \includegraphics[width=0.48\textwidth]{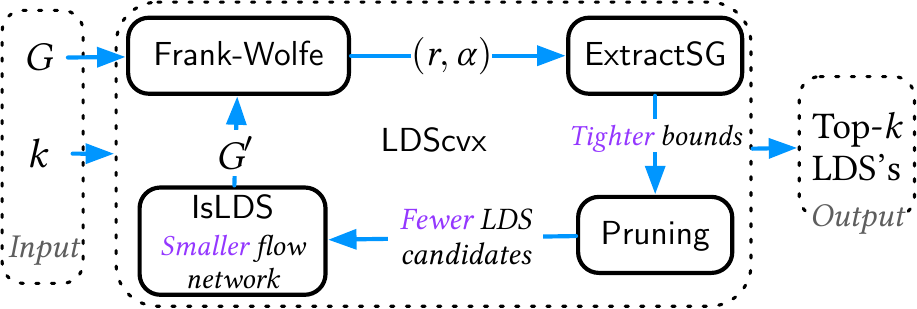}
    \caption{Algorithm workflow.}
    \label{fig:algo}
\end{figure}


\subsection{Extract Stable Groups}
\label{sec:algo:ext-sg}

In this subsection, we first introduce a new concept {\em stable group}, which can be used to provide upper and lower bounds of compact numbers, inspired by \MD{the stable subset in} \cite{danisch2017large}. Then, we discuss how to extract stable groups from the approximate solution $(\bm{r}, \bm{\alpha})$ provided by \FW.

\begin{definition}[Stable group]
    \label{def:stable-group}
    Given a feasible solution $(\bm{r}, \bm{\alpha})$ to $\mathsf{CP}(G)$, a stable group with respect to $(\bm{r}, \bm{\alpha})$ is a non-empty subset $S \in V$, if the following conditions hold.
    \begin{enumerate}
        \item For any $v \in V\setminus S$, $r_{v}$ satisfies either $r_{v} > \max_{u\in S}{r_{u}}$ or $r_{v} < \min_{u\in S}{r_{u}}$;
        \item For any $v \in V$, if $r_{v} > \max_{u\in S}{r_{u}}$, we have that $\forall (v,u) \in E\cap(\{v\}\times S)$, $\alpha_{v,u}=0$;
        \item For any $v \in V$, if $r_{v} < \min_{u\in S}{r_{u}}$, we have that $\forall (u,v) \in E\cap(S\times \{v\})$, $\alpha_{u,v}=0$.
    \end{enumerate}
\end{definition}

\begin{figure}[htb]
    \includegraphics[width=0.48\textwidth]{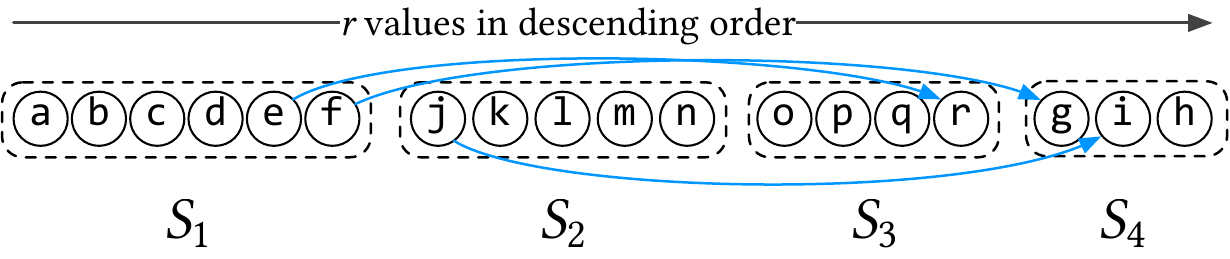}
    \caption{Stable groups w.r.t. $(\bm{r}^{*}, \bm{\alpha}^{*})$.}
    \label{fig:stable-groups}
\end{figure}

\Cref{def:stable-group} defines stable group.
We further use the stable groups w.r.t. $(\bm{r}^{*}, \bm{\alpha}^{*})$ (depicted in \Cref{fig:stable-groups}) as an example to illustrate the properties of the stable groups.
The definition indicates that if we sort all vertices $u\in V$ w.r.t. their $r$ values in descending order, we can observe the following properties:
\begin{enumerate}[leftmargin=0cm,itemindent=.5cm,labelwidth=\itemindent,labelsep=0cm,align=left]
    \item The vertices within the same stable group $S$ form a consecutive subsequence of the whole sequence. For example, the stable groups in \Cref{fig:stable-groups} give a partition to the entire sequence.
    \item The weights of edges whose endpoints fall into different stable groups are assigned to the endpoints with smaller $r$ values. In \Cref{fig:stable-groups}, we use arrows to denote weight assignments of edges across different stable groups. Note that edges within the same stable group are omitted. We can find that all arrows are pointed to the vertices with smaller $r$ values.
\end{enumerate}

Before discussing how to extract stable groups from $(\bm{r}, \bm{\alpha})$, we first prove by the following lemma that the stable groups help derive the upper and lower bounds of compact numbers.

\begin{lemma}
    \label{lem:phi-bounds}
    Given a feasible solution $(\bm{r}, \bm{\alpha})$ to \cref{equ:cp} and a stable group $S$ w.r.t. $(\bm{r}, \bm{\alpha})$, for all $u\in S$, we have that $\min_{v\in S}{r_{v}}\leq \phi(u) \leq \max_{v\in S}{r_{v}}$.
\end{lemma}

\begin{proof}
    We prove the lemma by contradiction. According to \Cref{thm:phi-r}, $\forall u \in V$, $\phi(u)=r_{u}^{*}$. Suppose there exists a vertex $u \in S$ such that $r_{u}^{*}=\phi(u) < \min_{v\in S}{r_{v}} \leq r_{u}$. There must exist another vertex $x\in V$ such that $r_{x}^{*} = \phi(x) > r_{x}$. Here $r_{x} \geq \min_{v\in S}r_{v}$ according to the 3-rd condition in \Cref{def:stable-group}. There exists $\epsilon > 0$ such that we could increase $r_{u}^{*}$ by $\epsilon$ and decrease $r_{x}^{*}$ by $\epsilon$, via manipulating the corresponding $\bm{\alpha}^{*}$ values, to strictly decrease $\Vert\bm{r}^{*}\Vert_{2}^{2}$ (i.e., the objective function). This contradicts that $\bm{r}^{*}$ is the optimal solution to $\mathsf{CP}(G)$.
\end{proof}
According to \Cref{lem:phi-bounds}, the minimum and maximum $r$ values in the stable group $S$ are the lower and upper bounds of compact numbers of vertices in $S$, respectively.

Now the key issue becomes how to extract stable groups from the approximate solution $(\bm{r}, \bm{\alpha})$ provided by \FW.
\MD{Our stable groups closely connect to the stable subsets in \cite{danisch2017large}.
    Compared to our stable group, the stable subset does not allow vertices not in the subset having larger $r$ values than vertices in the subset.
    For example, $S_{2}$ cannot be a stable subset, but $S_{1}\cup S_{2}$ is a stable subset.
    Hence, the stable subset \cite{danisch2017large} can be treated as the union of several stable groups with largest $r$ values.
    Due to the close relationship between our stable groups and stable subsets in \cite{danisch2017large}, we adapt the stable subset extraction method \cite{danisch2017large} to extract our stable groups.
    In general,}
we first extract the tentative stable groups from $\bm{r}$, and then verify or merge the tentative ones via \Cref{def:stable-group} to give the stable groups. For the optimal solution $(\bm{r}^{*}, \bm{\alpha}^{*})$ to $\mathsf{CP}(G)$, we can just group the vertices with the same $r$ values to form the stable groups (refer \Cref{tab:r-alpha-values} and \Cref{fig:stable-groups}). Although we cannot perform such aggregation based on $(\bm{r}, \bm{\alpha})$, the stable groups obtained from $(\bm{r}^{*}, \bm{\alpha}^{*})$ can still provide useful heuristic for us.

Consider the stable groups in \Cref{fig:stable-groups}. After sorting the vertices according to their $r$ values, let $V_{[1:i]}$ denote the first $i$ vertices in the sequence. We can find that the index $i$ of the last vertex in each stable group is exactly $\arg\max_{j\geq i}\mathsf{density}(G[V_{[1:j]}])$, i.e., the subgraph induced by more vertices than the first $i$ vertices cannot have a larger density. For example, the index of the last element in $S_{1}$ is 6, and the density of $G[V_{[1:6]}]$ is 3 and is the maximum among all subgraphs induced by $G[V_{[1:j]}]$, where $j \geq 6$. Hence, we use $\arg\max_{j\geq i}\mathsf{density}(G[V_{[1:j]}])$ to find indices for extracting stable group candidates. \MD{Note we break the tie via taking a larger index value for $j$, when two index values give the same density.} Then, we verify each candidate by \Cref{def:stable-group}. Specifically, we restrict all edges adjacent to the candidate stable group satisfying conditions (2) and (3) of \Cref{def:stable-group} by modifying $\bm{\alpha}$ and $\bm{r}$ and then check whether condition (1) of \Cref{def:stable-group} is fulfilled. If so, the candidate becomes a stable group. Otherwise, it may need to be merged with the next candidate.

\begin{algorithm}[tb]
    \caption{Extract stable groups from $(\bm{r}, \bm{\alpha})$}
    \label{alg:ext-sg}
    \algofont
    \Fn{\ExtSG{$G=(V,E)$, $\bm{r}$, $\bm{\alpha}$}}{
    sort vertices in $V$ according to $\bm{r}$: $r_{u_{1}} \geq r_{u_{2}} \geq \cdot \geq r_{u_{n}}$\;
    $I \gets \{i| i=\arg\max_{i\leq j \leq n} \mathsf{density}(G[V_{[1:j]}])\}$\;
    $\hat{\mathcal{S}} \gets$ partition $V$ according to $I$\;
    $\mathcal{S} \gets \emptyset$, $S\gets \emptyset$\;
    \While{$\hat{\mathcal{S}}$ is not empty}{
    $S' \gets$ pop out the first candidate from $\hat{\mathcal{S}}$\;
    $S \gets S \cup S'$\;
    {\color{teal}{\tcp{via \Cref{def:stable-group}}}}
    \If{$S$ is a stable group}{put $S$ into $\mathcal{S}$, $S \gets \emptyset$\;}
    }
    \ForEach{$S\in \mathcal{S}$}{
        \lForEach{$u \in S$}{$\overline{\phi}(u) \gets \min\{\overline{\phi}(u), \max_{v\in S} r_{v}\}$ }
        \lForEach{$u \in S$}{$\underline{\phi}(u) \gets \max\{\underline{\phi}(u), \min_{v\in S} r_{v}\}$ }
    }
    \Return{$\mathcal{S}$, $\overline{\phi}$, $\underline{\phi}$}
    }
\end{algorithm}

Based on the above discussion, we present the algorithm to extract stable groups from $(\bm{r}, \bm{\alpha})$, named \ExtSG, in \Cref{alg:ext-sg}. \ExtSG first sorts the vertices in $V$ according to their $r$ values descendingly (line 2). Then, \ExtSG finds the indices $I$ and extracts stable group candidates $\hat{\mathcal{S}}$ following the above heuristic (lines 3--4). Next, we check the candidate in $\hat{\mathcal{S}}$ one-by-one via \Cref{def:stable-group} (lines 6--10): if the candidate is a stable group, push it into the list of stable groups $\mathcal{S}$ (lines 9--10); otherwise, the current candidate $S$ will be merged with the next candidate $S'$ in the next iteration (line 8). After all stable groups in $\mathcal{S}$ are obtained, we update the upper and lower bounds of compact numbers according to \Cref{lem:phi-bounds} (lines 11--14). Finally, \ExtSG returns the stable groups $\mathcal{S}$ and updated upper and lower bounds of compact numbers (line 14). 


\subsection{Prune Invalid Vertices}
\label{sec:algo:pruning}

In this subsection, we present how to prune the vertices, which are certainly not contained by any LDS, based on compact number bounds derived in \Cref{sec:algo:ext-sg}.

We begin with a powerful corollary.

\begin{corollary}[Pruning rule 1]
    \label{cor:uv-prune}
    For any $u \in V$, if \ $\exists (u,v)\in E$, such that $\underline{\phi}(v) > \overline{\phi}(u)$, $u$ is not contained by any LDS in $G$.
\end{corollary}

\begin{proof}
    The corollary directly follows \Cref{lem:lds-out-phi}.
\end{proof}

\begin{example}[Pruning rule 1]
    \label{eg:uv-prune}
    Reconsider the graph $G$ in \Cref{fig:graph} and the stable groups in \Cref{fig:stable-groups}. For vertices in $S_{3}$, we can prune {\tt r}, as shown in \Cref{fig:pruning}. Because for edge $(\mathtt{e}, \mathtt{r})$, we have $\overline{\phi}(\mathtt{r})=\frac{3}{2} < \underline{\phi}(\mathtt{e})=\frac{5}{2}$, respectively. Similarly, the two vertices {\tt g} and {\tt i} in $S_{4}$ are also pruned by Pruning rule 1 (\Cref{cor:uv-prune}).
\end{example}

\begin{figure}[htb]
    \includegraphics[width=0.4\textwidth]{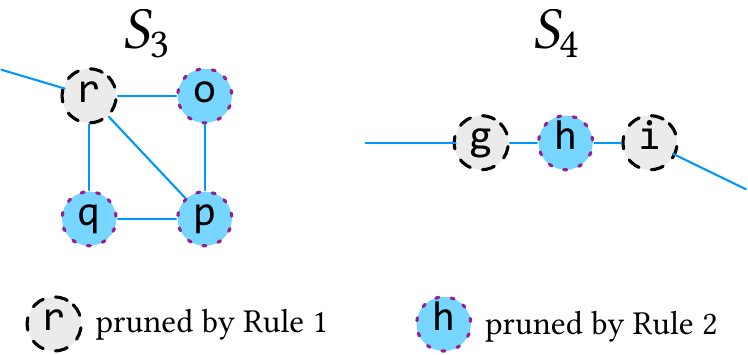}
    \caption{Pruning rules illustration.}
    \label{fig:pruning}
\end{figure}

Following \Cref{eg:uv-prune}, we can find that after removing {\tt r} and {\tt q} from $G$, denoting the graph after removal by $G'$, there is only one edge adjacent to {\tt o} and {\tt q}, as shown in \Cref{fig:pruning}. We can find that $\overline{\phi}_{G'}(\mathtt{o})=\overline{\phi}_{G'}(\mathtt{q})=1 < \underline{\phi}(\mathtt{o})=\underline{\phi}(\mathtt{q})=\frac{3}{2}$. It means that any LDS in $G$ cannot contain  {\tt o} and {\tt q}, because it needs to include some already pruned vertices, such as {\tt r}, to obtain a $\underline{\phi}(\mathtt{o})$-compact subgraph containing {\tt o} or {\tt q} in $G$. Hence, we derive another pruning rule from the above case.

\begin{lemma}[Pruning rule 2]
    \label{lem:uu-prune}
    Let $G'$ denote the graph after pruning vertices according to \Cref{cor:uv-prune} and \MD{\Cref{lem:uu-prune}}. For any vertex $u$ in $G'$, if $\overline{\phi}_{G'}(u) < \underline{\phi}(u)$, $u$ is not contained by any LDS in $G$.
\end{lemma}

\begin{proof}
    $\overline{\phi}_{G'}(u) < \underline{\phi}(u)$ means that only relying on the vertices in $G'$ cannot form a $\underline{\phi}(u)$-compact subgraph containing $u$, which means that some already pruned vertices are needed. Hence, $u$ cannot be contained by any LDS in $G$.
\end{proof}

To efficiently compute $\overline{\phi}_{G'}(u)$ in $G'$, we use $k$-core \cite{seidman1983network}, which is a cohesive subgraph model, following \cite{qin2015locally}.

\begin{definition}[$k$-core and core number \cite{seidman1983network}]
    The $k$-core of $G$ is the maximal subgraph $G[S]$ such that for any $u\in S$, $d_{G[S]}(u) \geq r$. For any $u \in V$, the {\em core number} of $u$, denoted by $\mathsf{core}_{G}(u)$, is the largest $k$ such that $u$ is contained in the $k$-core of $G$.
\end{definition}

\begin{lemma}\label{lem:upper-bound-cp-updated}
    Let $G'$ denote the graph after pruning invalid vertices. $\mathsf{core}_{G'}(u)$ provides an upper bound of $\phi_{G'}(u)$.
\end{lemma}

\begin{proofsketch}
    The lemma follows from Lemma 4.7 of \cite{qin2015locally}.
\end{proofsketch}

Following the above discussion about \Cref{eg:uv-prune}, \Cref{lem:upper-bound-cp-updated} provides a useful approach to obtain the upper bounds of compact numbers of {\tt o}  and {\tt p} after {\tt r} and {\tt q} are removed.

\begin{example}[Pruning rule 2]
    After {\tt r} is pruned from $G$ in \Cref{eg:uv-prune}, we obtain the upper bounds of compact numbers of {\tt o}, {\tt q}, and {\tt p} in the residual graph $G'$ via \Cref{lem:upper-bound-cp-updated}: $\overline{\phi}_{G'}(\mathtt{o})=\overline{\phi}_{G'}(\mathtt{q})=\overline{\phi}_{G'}(\mathtt{p})=1$. Then, we apply Pruning rule 2 (\Cref{lem:uu-prune}) to remove {\tt o}, {\tt q}, and {\tt p} from the graph, as shown in \Cref{fig:pruning}. Analogically, {\tt h} in $S_{4}$ is also pruned.
\end{example}

Following the above two examples, we further compare our pruning rules with those in {\tt LDSflow} \cite{qin2015locally}. {\tt LDSflow} mainly used core numbers for pruning: for vertex $u$, if $(u,v)\in E$ and $\mathsf{core}_{G}(u)<\frac{\mathsf{core}_{G}(v)}{2}$, or $\mathsf{core}_{G'}(u)<\frac{\mathsf{core}_{G}(u)}{2}$, $u$ can be pruned, where $G'$ denotes the graph with some vertices already pruned. From the perspective of compact numbers, the rationale behind the pruning in {\tt LDSflow} is that they actually used core numbers to provide relatively loose upper and lower bounds for compact numbers.

\begin{algorithm}[htb]
    \caption{Prune invalid vertices}
    \label{alg:pruning}
    \algofont
    \Fn{\Pruning{$G=(V,E)$, $\mathcal{S}$, $\overline{\phi}$, $\underline{\phi}$}}{
        $G'=(V', E') \gets G$\;
        \ForEach{$(u, v)\in E$}{\lIf{$\overline{\phi}(u)<\underline{\phi}(v)$}{remove $u$ from $G'$}}
        compute $\mathsf{core}_{G'}(u)$ for all vertices in $G'$\;
        \While{$\exists u \in V'$, $\mathsf{core}_{G'}(u) < \underline{\phi}(u) $}{
            remove $u$ from $G'$\;
            update core numbers of vertices adjacent to $u$\;
        }
        \lForEach{stable group $S \in \mathcal{S}$}{$S \gets S \cap V'$}
        \Return{$\mathcal{S}$\;}
    }
\end{algorithm}

Based on the two pruning rules, i.e., \Cref{cor:uv-prune,lem:uu-prune}, we present our pruning algorithm, named {\tt Pruning}, in \Cref{alg:pruning}. We first replicate the graph $G$ to $G'$ (line 1). Then, we apply Pruning rule 1 (\Cref{cor:uv-prune}) to remove invalid vertices (lines 3--4). Next, we compute the core numbers for all vertices in $G'$ (line 5). Afterwards, Pruning rule 2 (\Cref{lem:uu-prune}) is applied (lines 6--8). Finally, \Pruning updates the stable groups by intersecting them with the vertices not been pruned (line 9), and returns the updated stable groups (line 10).

The following subsection will introduce how to extract and verify LDS from the updated stable groups.


\subsection{Extract and Verify LDS}
\label{sec:algo:verify}

Here, we discuss how to extract and verify the LDS from the stable groups after pruning. First, we can find that the vertices within the stable group $S$ with largest $r$ values in $\mathcal{S}$ satisfy \Cref{lem:lds-out-phi} w.r.t. the current valid vertices, otherwise they are pruned in \Pruning. But, we are not sure whether $G[S]$ is the densest among all its subgraphs (i.e., $G[S]$ is self-densest) and whether there exists another larger $\mathsf{density}(G[S])$-compact subgraph of $G$ containing $G[S]$, according to \Cref{def:lds}.

We first examine whether $G[S]$ is self-densest because the computation cost for self-densest examination is smaller than checking whether it is a maximal $\mathsf{density}(G[S])$-compact subgraph.

\MD{Verifying whether $G[S]$ is the densest among all subgraphs of $G[S]$ is one step in the binary search of computing densest subgraph \cite{goldberg1984finding}, i.e., checking whether there is a subgraph with higher density than $\mathsf{density}(G[S])$. Generally, we use {\tt IsDensest} to verify the self-densest via computing the max-flow on the flow network generated based on $G[S]$ following \cite{sun2020kclist++}.
}

If $G[S]$ is the DS of itself (i.e., {\tt IsDensest} returns {\sf True}), we need to further verify whether $G[S]$ is the maximal $\mathsf{density}(G[S])$-compact subgraph in $G$. We first review how $G[S]$ is verified as an LDS in \cite{qin2015locally}, and next we give our improved verification algorithm.

\citeauthor{qin2015locally} \cite{qin2015locally} first use breadth-first-search starting from $G[S]$ to traverse each vertice $u$ with $\overline{\phi}(u) \geq \mathsf{density}(G[S])$. Recall that they use $\mathsf{core}_{G}(u)$ as $\overline{\phi}(u)$ (briefed in \Cref{sec:algo:pruning}). We use $G^{t}$ to denote the subgraph traversed. If there does not exist an already computed LDS in $G^{t}$, $G[S]$ is an LDS. Otherwise, \citeauthor{qin2015locally} construct a flow network based on $G^{t}$, then use the min-cut algorithm to find all maximal $\mathsf{density}(G[S])$-compact subgraphs in $G^{t}$, and check whether $G[S]$ is maximal.

We can observe that the verification algorithm in \cite{qin2015locally} needs to compute the min-cut on the flow-network based on the vertices with $\overline{\phi}(u) \geq \mathsf{density}(G[S])$. We will show that only the vertices with $\overline{\phi}(u) \geq \mathsf{density}(G[S])$ and $\underline{\phi}(u) \leq \mathsf{density}(G[S])$, which form a subset of the set of vertices needed in \cite{qin2015locally}, are needed to verify whether $G[S]$ is an LDS of $G$.

\Cref{alg:is-lds} presents our improved verification algorithm, named \IsLDS. \IsLDS first initializes an empty queue $Q$, an empty vertex set $U$, an empty edge set $L$, and $\rho$ with $\mathsf{density}(G[S])$ (lines 2--3). Next, the algorithm performs a breadth-first search starting from $S$ (lines 4--13). Specifically, \IsLDS uses $Q$ to store the vertices to be traversed. Each time, it pops out the first vertex $v$ from $Q$ (line 5), and iterates all neighbors of $v$ (lines 8--13). For each neighbor $w$, if $\underline{\phi}(w) > \rho$, $w$ will not be added to $U$ and $Q$, but a self loop of $v$ is added to $L$ (lines 10--11). If $\underline{\phi}(w) \leq \rho \leq \overline{\phi}(w)$, $w$ is added into $Q$ and $U$ (lines 12--13).
If \IsLDS does not encounter a vertex with $\underline{\phi}(w) > \rho$ during the traversal, we can return {\sf True} (line 14), which means that there does not exist an already computed LDS in the traversed subgraph.
Otherwise, we construct a subgraph $G^{t}$ with all edges induced by $U$ and self loops in $L$ (lines 15). Afterward, we compute all $\rho$-compact subgraphs in $G^{t}$ via min-cut following \cite{qin2015locally} (line 16). Finally, we return {\sf True} if $G[S]$ is maximal $\rho$-compact; otherwise, {\sf False} is returned (line 17).
We can observe that $G^{t}$ only contains vertices with $\overline{\phi}(w) \geq \mathsf{density}(G[S]) \geq \underline{\phi}(w)$. Hence, the flow network generated in our algorithm is much smaller than that generated in \cite{qin2015locally}.

\begin{algorithm}[tb]
	\caption{Check whether $G[S]$ is an LDS of $G$}
	\label{alg:is-lds}
	\algofont
	\Fn{\IsLDS{$S$, $\overline{\phi}$, $\underline{\phi}$, $G=(V,E)$}}{
		$Q \gets $ an empty queue, $\rho \gets \mathsf{density}(G[S]) $\;
		$U \gets \emptyset$, $L \gets \emptyset$, $needFlow\gets \mathsf{False}$\;
		\ForEach{$u \in S$} {
			\lIf{$u \notin U$ }{push $u$ to $Q$, insert $u$ into $U$}
			\While{$Q$ is not empty}{
				$v \gets$ pop out the front vertex in $Q$\;
				\ForEach{$(v,w) \in E$}{
					\If{$w \notin U$ }{
						\If{$\underline{\phi}(w)>\rho$}{add edge $(v, v)$ to $L$, $needFlow\gets \mathsf{True}$\;}
						\ElseIf{$\overline{\phi}(w)>\rho$}{push $w$ to $Q$, add $w$ into $U$\;}
					}
				}
			}
		}
		\lIf{not $needFlow$}{\Return{{\sf True}}}
		$G^{t} \gets (U, E(U)\cup L)$\;
		$G' \gets$ all $\rho$-compact subgraphs in $G^{t}$ via min-cut\;
		\Return{$G[S]$ is a connected component in $G'$}\;
	}

\end{algorithm}

Before proving the correctness of \Cref{alg:is-lds}, we use an example to illustrate the traversed subgraph $G^{t}$.

\begin{figure}[ht]
	\includegraphics[width=0.46\textwidth]{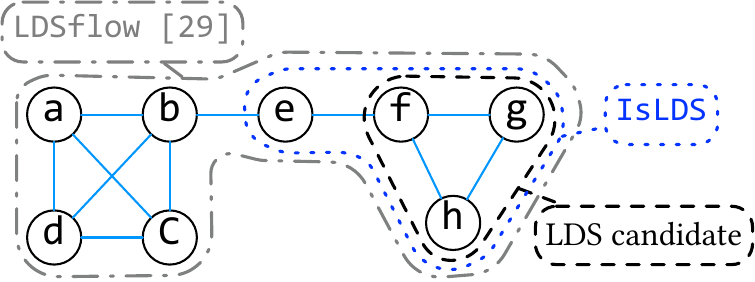}
	\caption{LDS verification illustration.}
	\label{fig:verify}
\end{figure}

\begin{example}
	Consider the graph in \Cref{fig:verify}. Suppose $S=\{\mathtt{f}, \mathsf{g}, \mathsf{h}\}$ and we want to verify whether $G[S]$ is an LDS of $G$. Clearly, $G[S]$ is the DS of itself. We illustrate the scope of the subgraph $G^{t}$ in \Cref{alg:is-lds}. Following \Cref{alg:is-lds}, $U$ contains {\tt f}, {\tt g}, {\tt h}, {\tt e} and $L$ consists of edge ({\tt e}, {\tt e}) because $\underline{\rho}(${\tt b}$) > 1=\mathsf{density}(G[S])$. Hence, $G^{t}$ in our \IsLDS contains 4 vertices and 5 edges, while the traversed subgraph in {\tt LDSflow} \cite{qin2015locally} contains all eight vertices (i.e., {\tt a}, {\tt b}, {\tt c}, {\tt d}, {\tt e}, {\tt f}, {\tt g}, {\tt h}) as shown in \Cref{fig:verify}, because core numbers of all vertices are larger than 1, i.e., $\mathsf{density}(G[S])$.
\end{example}


\begin{theorem}
	\label{thm:isLDS}
	Given a graph $G$ and a subgraph $G[S]$, where $G[S]$ is the DS of itself, $G[S]$ is an LDS of $G$ if and only if \Cref{alg:is-lds} returns {\sf True}.
\end{theorem}

\begin{proof}
	First, if $G[S]$ is an LDS, \Cref{alg:is-lds} returns {\sf True}. Because only the loops in $L$ might increase the compact numbers in $G^{t}$ compared to the compact numbers in $G$. Hence, $G[S]$ is still an LDS in $G^{t}$. Otherwise the maximal $\mathsf{density}(G[S])$-compact subgraph containing $G[S]$ must contain a vertex $u$ with self loop, and then we can construct a larger $\mathsf{density}(G[S])$-compact subgraph in $G$ by including vertices with $\underline{\phi}(w) > \mathsf{density}(G[S])$ connected to $u$. Hence, the contradiction proves the claim.

	On the other direction, if $G[S]$ is not an LDS of $G$, we will find a larger $\mathsf{density}(G[S])$-compact subgraph containing $G[S]$ in $G^{t}$. Hence, $G[S]$ is also not an LDS in $G^{t}$. Thus, the algorithm returns {\sf False}.
\end{proof}

By now, we have introduced all building blocks of our LDS algorithm. In the next subsection, we will present our LDS algorithm {\tt LDScvx} by combining these components.

\subsection{The Overall Algorithm {\tt LDScvx}}
\label{sec:algo:lds}

Combining \Cref{alg:fw,alg:pruning,alg:ext-sg,alg:is-lds} with reference to \Cref{fig:algo}, we will obtain our LDS algorithm, named {\tt LDScvx} in \Cref{alg:LDScvx}.

\begin{algorithm}[tb]
    \caption{Our LDS algorithm, {\tt LDScvx}}
    \label{alg:LDScvx}
    \algofont
    \Input{A graph $G=(V,E)$ and two integers $k$ and $N$}
    \Output{LDSs with top-$k$ densities}
    preprocess $G$ by pruning vertices via core numbers\;
    $G' \gets G$\;
    $stk \gets$ an empty stack\;
    \While{$k>0$} {
        $(\bm{r},\bm{\alpha}) \gets$ \FW{$G'$, $N$}\;
        $\mathcal{S}, \overline{\phi}, \underline{\phi} \gets$ \ExtSG{$G'$, $\bm{r}$,$\bm{\alpha}$}\;
        $\mathcal{S} \gets$ \Pruning{$G'$, $\mathcal{S}$, $\overline{\phi}$, $\underline{\phi}$}\;
        \lForEach{$S \in \mathcal{S}$ reversely}{push $S$ into $stk$}
        $S \gets $ pop out the top stable group from $stk$\;
        \If{\IsDensest{$G[S]$}}{
            \lIf{\IsLDS{$S$, $\overline{\phi}$, $\underline{\phi}$, $G$}}{output $G[S]$, $k\gets k-1$}
            \lIf{$stk$ is empty}{break}
            $ S\gets $ pop out the top stable group from $stk$\;
        }
        $G'\gets G[S]$\;
    }
\end{algorithm}

In {\tt LDScvx}, we first assign $G$ to $G'$ (line 1) and initialize an empty stack $stk$ (line 2). Next, we extract the stable groups from the graph $G'$ via \FW, \ExtSG, and \Pruning (lines 4-6). Then, the algorithm pushes the stable groups in $\mathcal{S}$ reversely into $stk$ (line 7). For stable groups in $stk$, the corresponding $\underline{\phi}$ value is decremented from top to bottom. Afterward, the first stable group in $stk$, which is also the one with the highest $\underline{\phi}$ value, is poped out (line 8) and is examined by \IsDensest and \IsLDS (line 9--10). If $G[S]$ is an LDS, we output it and decrease $k$ by 1 (line 10). If $G[S]$ is not an LDS but is the DS of itself, we update $S$ as the top stable group from $stk$ (line 12). Next, we assign $G[S]$ to $G'$ for the next iteration (line 13). The above process is repeated until top-$k$ LDSs are found (line 3), or the stack is empty (line 11).

Next, we use an example to explain further the overall procedure of {\tt LDScvx} (\Cref{alg:LDScvx}).
\begin{example}
    We still use the graph $G$ in \Cref{fig:graph} as the example. Suppose we want to find top-2 LDSs from $G$. Assume we obtain $(\bm{r}^{*}, \bm{\alpha}^{*})$ in \Cref{tab:r-alpha-values} after \FW. Then, we will obtain the stable groups shown in \Cref{fig:stable-groups}, as well as the upper and lower bounds of compact numbers via \ExtSG. Next, in the \Pruning process, the vertices in $S_{3}$ and $S_{4}$ will be pruned according to \Cref{cor:uv-prune} and \Cref{lem:uu-prune}, which means that $\mathcal{S}$ contains $S_{1}$ and $S_{2}$. Afterward, $S_{2}$ and $S_{1}$ will be pushed into the stack $stk$.
    Now, the stable groups in $stk$ satisfy that the compact numbers of vertices in the stable group higher in $stk$ are larger than those in the stable group lower in $stk$.
    Next, we pop out the top stable group $S_{1}$ from $stk$ and verify that it is an LDS via \IsDensest and \IsLDS. We output $G[S_{1}]$ as the first LDS, pop out $S_{2}$ from $stk$, and repeat the above process. After $S_{2}$ is verified as an LDS, $stk$ is empty, we break while loop (line 10 in \Cref{alg:LDScvx}). In the end, we obtain two LDSs, $G[S_{1}]$ and $G[S_{2}]$.
\end{example}

{\bf Complexity.} The time complexity of {\tt LDScvx} is $O((N_{\mathsf{FW}}+N_{\mathsf{SG}})\cdot (n+m) + N_{\mathsf{Flow}}\cdot t_{\mathsf{Flow}})$, where $N_{\mathsf{FW}}$ is number of iterations that \FW needs, and $N_{\mathsf{SG}} \leq n$ is the number of stable groups in total, $N_{\mathsf{Flow}}$ is number of times \IsLDS and \IsDensest are called, and $t_{\mathsf{Flow}}$ denotes the time complexity of max-flow computation. Note that an iteration in \FW and verifying a stable group in \ExtSG both take $O(n+m)$ time cost. The memory complexity is $O(n+m)$.


{\color{blackkkk} \section{The LTDS Problem and Our Solution}
\label{sec:ltds}
\citeauthor{samusevich2016local} \cite{samusevich2016local} extended the LDS model from edge-based density to triangle-based density, termed locally triangle-densest subgraph (LTDS), and proposed a max-flow-based solution, named \texttt{LTDSflow}. We now study the LTDS problem and show how our previous {\tt LDScvx} can be adapted.

\subsection{The LTDS Problem}
Following the classic definition \cite{tsourakakis2014novel,tsourakakis2015k,fang2019efficient,sun2020kclist++,zhou2024counting}, the triangle-based density of a graph $G = (V, E)$, denoted by $\mathsf{tr}\text{-}\mathsf{density}(G)$, is defined as:
\begin{equation}
    \mathsf{tr}\text{-}\mathsf{density}(G) = \frac{|T|}{|V|},
    \label{eq:density}
\end{equation}
where $T$ is the set of triangles in $G$.
Based on the triangle-based density, we formally introduce the definition of LTDS below.\\

\begin{definition}
    [$\rho$-tr-compact \cite{samusevich2016local}]
    \label{def:compact}
    A graph $G = (V, E)$ is $\rho$-tr-compact if and only if $G$ is connected, and removing any subset of vertices $S \subseteq V$ will result in the removal of at least $\rho \times |S|$ triangles in $G$, where $\rho$ is a nonnegative real number.
\end{definition}

\begin{definition}
    [Maximal $\rho$-tr-compact subgraph \cite{samusevich2016local}]
    \label{def:maximal}
    A $\rho$-tr-compact subgraph $G[S]$ of $G$ is a maximal $\rho$-tr-compact subgraph of $G$ if and only if there does not exist a supergraph $G[S']$ of $G[S]$ with $S \subset S'$ in $G$ such that $G[S']$ is also $\rho$-tr-compact.
\end{definition}

\begin{definition}
    [Locally triangle-densest subgraph \cite{samusevich2016local}]
    \label{def:ltds}
    A subgraph $G[S]$ of $G$ is a locally triangle-densest subgraph (LTDS) of $G$ if and only if $G[S]$ is a maximal $\mathsf{tr\text{-}density}(G[S])$-tr-compact subgraph in $G$.
\end{definition}

\Cref{def:compact,def:maximal,def:ltds} naturally follow from \Cref{def:rho-compact,def:maximal-rho,def:lds}, indicating that LTDS shares many properties with LDS. We list some important properties in \Cref{lem:ltds}. This lemma ensures that LTDS is indeed the triangle-densest in its local region.
\begin{lemma}\label{lem:ltds}
    Given a graph $G=(V,E)$, an LTDS $G[S]$ of G has following properties:
    \begin{enumerate}
        \item any subgraph $ G[S'] $ of an LTDS $ G[S] $ can not have a larger tr-density than $G[S]$;
        \item any supergraph $ G[S'] $ of an LTDS $ G[S] $, $ G[S'] $ is not $ \rho $-tr-compact for any $ \rho \geq \mathsf{tr\text{-}density}(G[S]) $.
    \end{enumerate}
\end{lemma}
\begin{proof}
    Please refer to the appendix.
\end{proof}

\begin{problem}
[LTDS problem \cite{samusevich2016local}]
Given a graph $G$ and an integer $k$, the LTDS problem is to compute the top-$k$ LTDS's with the largest density in $G$.
\end{problem}

\begin{example}[LTDS]
    Consider the graph $G$ shown in \Cref{fig:graph}. The subgraph $G[S_{1}]$ with tr-density $\frac{10}{3}$ is a maximal $\frac{10}{3}$-tr-compact subgraph. Hence, $G[S_{1}]$ is an LTDS. Similarly, $G[S_{2}]$ with tr-density $2$ is also an LTDS as it is a maximal $2$-tr-compact subgraph.
    The subgraph $G[S_{3}]$ with tr-density $\frac{1}{2}$ is a $\frac{1}{2}$-tr-compact subgraph. But $G[S_{3}]$ is not an LTDS because it is contained in $G[S_{1}\cup S_{3}]$ which is also $\frac{1}{2}$-tr-compact.
    $G[S_{1}\cup S_{3}]$ is also not an LTDS, because its tr-density is $\frac{11}{5}$ but it is not a $\frac{11}{5}$-tr-compact subgraph. The tr-compactness of $G[S_{1}\cup S_{3}]$ is $\frac{1}{2}=\frac{2}{4}$, because removing $S_{3}$ from $G[S_{1}\cup S_{3}]$ will result in removing of 2 triangles.
\end{example}

\subsection{Our LTDS Algorithm}
\textbf{Frank-Wolfe-based algorithm.} \Cref{equ:trcp} presents the convex program for triangle-densest subgraphs \cite{sun2020kclist++}. This formulation is a natural extension of \Cref{equ:cp}. The key difference between \Cref{equ:trcp} and \Cref{equ:cp} lies in their objectives: the former focuses on distributing the weights of triangles, whereas the latter focuses on the weights of edges. Similarly, the Frank-Wolfe-based algorithm for \Cref{equ:trcp} is also a straightforward generalization of \Cref{alg:fw}. The pseudo-code is presented in \Cref{alg:trfw}.
\begin{equation}\label{equ:trcp}
    \begin{aligned}
        \mathsf{CP}(G) &  & \min\                             & \sum\limits_{u\in V}r_u^2                       &                                  \\
                       &  & r_u                               & = \sum\limits_{u \in t, t\in T } \alpha_{u, t}, & \forall u \in V                  \\
                       &  & \sum\limits_{u \in t}\alpha_{u,t} & \ge 1,                                          & \forall t \in T                  \\
                       &  & \alpha_{u, t}                     & \ge 0,                                          & \forall u \in V, \forall t \in T \\
    \end{aligned}
\end{equation}

\begin{algorithm}[t]
    \caption{Frank-Wolfe-based algorithm \cite{danisch2017large}}
    \label{alg:trfw}
    \algofont
    \Fn{\FW{$G=(V,E,T)$, $N\in \mathbb{Z}_{+}$}}{
        \lForEach{$t \in T$}{$\alpha_{u, t}^{(0)} \leftarrow \frac{1}{3},\forall u \in t$}
        \lForEach{$u \in V$}{$r^{(0)}_{u} \gets \sum_{t\in T \wedge u\in t} \alpha^{(0)}_{u,t}$}
        \For{$i=1, \dots, N$}{
            $\gamma_{i} = \frac{2}{i+2}$\;
            \ForEach{$u \in t$}{
                $\hat{\alpha}_{u,t}\leftarrow 1$ if $u=x$ and $0$ otherwise
            }
            $\bm{\alpha}^{(i)} \gets (1-\gamma_{i}) \cdot \bm{\alpha}^{(i-1)} + \gamma_{i} * \hat{\bm{\alpha}}$\;
            \lForEach{$u \in V$}{$r^{(i)}_{u} \gets \sum_{t\in T \wedge u\in t} \alpha^{(i)}_{u,t}$}
        }
        \Return{$(\bm{r}^{(i)}, \bm{\alpha}^{(i)})$}\;
    }
\end{algorithm}

To bridge LTDS and \Cref{equ:trcp}, we introduce the definition of tr-compact number, generalized from \Cref{def:cn}.
\begin{definition}
    [Tr-compact number]
    \label{def:trcn}
    Given an undirected graph $G = (V, E)$, the tr-compact number of each vertex $u \in V$, denoted by $\phi_{tr}(u)$, is the largest $\rho$, such that $u$ is contained in a $\rho$-tr-compact subgraph of $G$.
\end{definition}

\begin{theorem}
    \label{th:cnltds}
    Any connected subgraph $G[S]$ is an LTDS in $G$ if and only if the following conditions are satisfied:
    \begin{enumerate}
        \item $\forall u \in S$, $\phi_{tr}(u) = \mathsf{tr\text{-}density}(G[S])$;
        \item $\forall (u, v) \in E$, such that $u \in S$ and $v \in V \setminus S$, it holds that $\phi_{tr}(u) > \phi_{tr}(v)$.
    \end{enumerate}
\end{theorem}
\begin{proof}
    Please refer to the appendix.\qed
\end{proof}

\begin{theorem}
    \label{th:cp}
    Suppose $({\bm{\mathbf{r^*}}}, \bm{\alpha}^*)$ is an optimal solution of \Cref{equ:trcp}. Then, each $r_u^*$ in $\bm{\mathbf{r^*}}$ is exactly the tr-compact number of $u$, i.e., $\forall u \in V$, $r_u^* = \phi_{tr}(u)$.
\end{theorem}
\begin{proof}
    Please refer to the appendix.\qed
\end{proof}

\Cref{th:cnltds} enables us to efficiently extract and verify LTDS's using tr-compact numbers, while \Cref{th:cp} ensures that \Cref{alg:trfw} can correctly compute feasible tr-compact numbers.

\noindent \textbf{Extract stable groups and pruning.} Inspired by the idea of $k$-core \cite{seidman1983network}, \citeauthor{samusevich2016local} proposed tr-$k$-core to optimize the pruning process of \texttt{LTDSflow}.
\begin{definition}[tr-$k$-core and tr-core number \cite{samusevich2016local}]
    The tr-$k$-core of $G$ is the maximal subgraph $G[S]$ such that for any $u\in S$, $u$ is contained in at least $k$ triangles. For any $u \in V$, the {\em tr-core number} of $u$, denoted by $\mathsf{tr\text{-}core}_{G}(u)$, is the largest $k$ such that $u$ is contained in the tr-$k$-core of $G$.
\end{definition}
We borrow the idea of tr-core numbers to compute the initial bounds for tr-compact numbers. After initialization, \Cref{def:stable-group} and \ExtSG can be directly adapted to the triangle-based setting by substituting edge weights with triangle weights and replacing edge-based density with triangle-based density.
\Cref{def:trsg} gives a generalized definition of \Cref{def:stable-group}.
\begin{definition}[Triangle-based stable groups]
    \label{def:trsg}
    Given a feasible solution $(\bm{r}, \bm{\alpha})$ to \Cref{equ:trcp}, stable group $S \in V$ with respect to $(\bm{r}, \bm{\alpha})$ is a non-empty subset if the following conditions hold:
    \begin{enumerate}
        \item For any $v \in V \setminus S$, $r_v$ satisfies either $r_v > \max_{u \in S} r_u$ or $r_v < \min_{u \in S} r_u$;
        \item For any $t \in T$, if $t \cap S \neq \emptyset$, $\sum_{r_u > \max_{v \in S} r_v} \alpha_{u, t} = 0$;
        \item For any $t \in T$, if $t \cap S \neq \emptyset$, $\sum_{r_u < \min_{v \in S} r_v} \alpha_{u, t} = 1$.
    \end{enumerate}
\end{definition}
The modified \ExtSG provides tighter bounds for tr-compact numbers compared to tr-core numbers \cite{samusevich2016local}. The correctness of the bounds is guaranteed by \Cref{lem:tr-phi-bounds}. Similarly, \Pruning can also be adapted.
\begin{lemma}
    \label{lem:tr-phi-bounds}
    Given a feasible solution $(\bm{r}, \bm{\alpha})$ to \Cref{equ:trcp} and a triangle-based stable group $S$ w.r.t. $(\bm{r}, \bm{\alpha})$, for all $u\in S$, we have that $\min_{v\in S}{r_{v}}\leq \phi_{tr}(u) \leq \max_{v\in S}{r_{v}}$.
\end{lemma}
\begin{proofsketch}
    We can prove the lemma by generalizing \Cref{lem:phi-bounds} for supporting triangles.
\end{proofsketch}

\noindent \textbf{Extract and verify LTDS.} To verify the candidates returned by \Pruning, we need to examine them from two perspectives based on the properties in \Cref{lem:ltds}. Since \IsDensest effectively verifies self-densest property, the only remaining task is to determine whether the candidate subgraph $G[S]$ is a maximal $\mathsf{tr\text{-}density}(G[S])$-tr-compact subgraph. 
To efficiently solve this problem, we propose a different flow network construction method, named \ExtMaximal, and a new LTDS algorithm, named \IsLTDS, based on the original \IsLDS, to address these issues.

\begin{algorithm}[tb]
    \caption{Check whether $G[S]$ is an LTDS of $G$}
    \label{alg:is-ltds}
    \algofont
    \Fn{\IsLTDS{$S$, $\overline{\phi}_{tr}$, $\underline{\phi}_{tr}$, $G=(V,E)$}}{
        $Q \gets $ an empty queue, $\rho \gets \mathsf{tr\text{-}density}(G[S]) $\;
        $U \gets \emptyset$, $P \gets \emptyset$, $T_{visited}\gets \emptyset$, $needFlow\gets \mathsf{False}$\;
        \ForEach{$u \in S$}{
            push $u$ to $Q$, insert $u$ into $U$\;
        }

        \While{$Q$ is not empty}{
            $v \gets$ pop out the front vertex in $Q$\;
            \ForEach{$t \in T\wedge v \in t\wedge t\notin T_{visited}$}{
                $valid\gets \textsf{True}$, $T_{visited}\gets T_{visited}\cup t$\;
                \ForEach{$w\in t$}{
                    \lIf{$\overline{\phi}_{tr}(w)< \rho$}{$valid\gets \textsf{False}$}
                }
                \lIf{not valid}{continue}
                $M\gets \emptyset$\;
                \ForEach{$w\in t$}
                {
                    \If{$\underline{\phi}_{tr}(w)\leq \rho$}{
                        add $w$ into $M$\;
                        \lIf{$w\notin U$}{
                            push $w$ to $Q$, add $w$ into $U$
                        }
                    }
                    \Else{
                        $needFlow\gets \mathsf{True}$\;
                    }
                }
                \lIf{$M\neq \emptyset$}{$P\gets P\cup \{M \}$}
            }
            \ForEach{$(v,w)\in E$}{
                \If{$w \notin U$ }{
                    \lIf{$\underline{\phi}_{tr}(w)>\rho$}{$needFlow\gets \mathsf{True}$}
                    \ElseIf{$\overline{\phi}_{tr}(w)>\rho$}{push $w$ to $Q$, add $w$ into $U$\;}
                }
            }
        }
        \lIf{not $needFlow$}{\Return{{\sf True}}}
        $S'\gets$\ExtMaximal($U$, $P$, $\rho$)\;
        \Return{$G[S]$ is a connected component in $G[S']$}
    }

\end{algorithm}

\begin{algorithm}[tb]
    \caption{Extract maximal $\rho$-tr-compact subgraphs}
    \label{alg:is-maximal}
    \algofont
    \Fn{\ExtMaximal{$S$, $P$, $\rho$}}{
        initialize a flow network $\mathcal{F}=(\mathcal{V},\mathcal{E})$ with a source vertex $s$ and a sink vertex $t$\;
        \ForEach{$u\in S$}{
            add a node $\mathsf{u}$ into $\mathcal{V}$\;
            add an edge $(\mathsf{u}, t)$ with capacity $\rho$ into $\mathcal{E}$\;
        }
        \ForEach{$p \in P$}{
            add a node $\mathsf{p}$ into $\mathcal{V}$\;
            add an edge $(s, \mathsf{p})$ with capacity $1$ into $\mathcal{E}$\;
            \ForEach{$u\in p$}{
                add an edge $(\mathsf{p},\mathsf{u})$ with capacity $1$ into $\mathcal{E}$\;
            }
        }
        $\mathcal{S},\mathcal{T}\gets$ minimal s-t cut in $\mathcal{F}$\;

        \Return{$(\mathcal{S}\cap S)\setminus {s}$;}
    }

\end{algorithm}

\Cref{alg:is-ltds} presents our verification process for LTDS.
From a high-level perspective, \IsLTDS performs a breadth-first search starting from the candidate subgraph $G[S]$, aiming to identify all vertices and triangles that could be part of a $\mathsf{tr\text{-}density}(G[S])$-tr-compact subgraph. Afterward, with the help of \ExtMaximal, we further verify whether $G[S]$ is maximal.

Specifically, an empty queue $Q$, an empty vertex set $U$, empty pattern sets $P$, $T_{visited}$ and tr-density $\rho$ are first initialized (lines 2-3).
Here, $U$ and $P$ are used to collect all vertices and triangles that would be included in a $\rho$-tr-compact subgraph, while $T_{visited}$ is used to avoid repeated triangle visits during the traversal.
We use an additional boolean parameter $needFlow$ to indicate whether the flow-based function \ExtMaximal should be executed.
 Then, we push each vertex $u\in S$ into $Q$ (lines 4-5). Next, we pop out the front vertex $v$ from $Q$ and begin our twofold verification (lines 7-25). In the first aspect, we search every unvisited triangle containing $v$ to find valid triangles. 
 A valid triangle consists of three vertices, each with an upper bound of the tr-compact number greater than or equal to  $\rho$.
 {\em Here, only if the triangle is valid, it can be part of the maximal $\mathsf{tr\text{-}density}(G[S])$-tr-compact subgraph, which helps reduce the search range and computational costs. }
For each valid triangle $t$, we extend the vertex set $U$ by adding vertices $w \in t$ that satisfy $\underline{\phi}_{tr}(w) \leq \rho$. 
A new vertex set $M$ is maintained to store these vertices. If $M$ is not empty, we add it to the pattern set $P$ to extend it.
{\em Vertices with $\underline{\phi}_{tr}(w) > \rho$ in $t$ are excluded here to speed up the maximality verification. The proof of correctness can be found in the appendix.}
In the second aspect, \IsLTDS iterates all neighbors of $v$ to further extend $U$ (lines 21-25).
If \IsLTDS does not encounter a vertex with a lower bound of the tr-compact number greater than  $\rho$, we directly return $\mathsf{True}$ (line 26).
Otherwise, we must construct a well-designed flow network and use the max-flow algorithm to find the maximal $\rho$-tr-compact subgraph from $G[U]$. Finally, \IsLTDS returns $\mathsf{True}$ if $G[S]$ is a connected component of $G[S']$.

\Cref{alg:is-maximal} describes our new flow network construction method. The intuition behind \ExtMaximal is to distribute the weights of each valid triangle to vertices in $U$, which is quite similar to the intuition behind \Cref{equ:trcp}. Specifically, each pattern $p\in P$ is connected to source vertex $s$ through an edge with capacity $1$, representing the weight of a triangle. At the same time, all vertices in $p$ are also connected to $p$ with a capacity of $1$. Moreover, each vertex is linked to the sink vertex $t$ with a capacity of $\rho$. The correctness of \ExtMaximal and \IsLTDS is stated in \Cref{lem:network} and \Cref{th:isltds}, respectively.

\begin{lemma}
    \label{lem:network}
    Given a tr-density value $\rho$, along with the corresponding vertex set $U$ and pattern set $P$ generated in lines 6-25 of \Cref{alg:is-ltds}, \Cref{alg:is-maximal} returns all vertices $u$ whose tr-compact number is larger than or equal to $\rho$.
\end{lemma}
\begin{proof}
    Please refer to the appendix.\qed
\end{proof}

\begin{theorem}
    \label{th:isltds}
    Given a graph $G$ and a subgraph $G[S]$, where $G[S]$ is the TDS of itself, $G[S]$ is an LTDS of $G$ if and only if \Cref{alg:is-ltds} returns {\sf True}.
\end{theorem}
\begin{proof}
    Please refer to the appendix.\qed
\end{proof}

\noindent \textbf{The overall algorithm {\tt LTDScvx}.} The workflow of {\tt LTDScvx} can be naturally extended from {\tt LDScvx}.}

{\color{blackkkk} \section{A Unified Framework for LDS Problem}
\label{sec:framework}
\citeauthor{zhou2024depth} \cite{zhou2024depth} proposed a unified framework for the densest subgraph discovery (DSD) problem. This framework consists of three stages: {\it Graph reduction}, {\it Vertex Weight Updating} ({\tt VWU}), and {\it Candidate Subgraph Extraction and Verification} ({\tt CSEV}). Inspired by this, we consolidate all LDS and LTDS solutions ({\tt LDSflow} \cite{qin2015locally}, {\tt LTDSflow} \cite{samusevich2016local}, {\tt LDScvx}, and {\tt LTDScvx}) into a similar framework.

\begin{algorithm}[tb]
    \caption{The unified framework.}
    \label{alg:framework}
    \algofont
    \Input{A graph $G=(V,E)$ and an integer $k$}
    \Output{Top-$k$ valid subgraph}
    $stk \gets$ an empty stack, $\bm{w}\gets \emptyset$, $\overline{\phi}\gets \emptyset$, $\underline{\phi}\gets \emptyset$\;
    \tcp{\textcolor{teal}{The initial reduction method.}}
    $V', E', \overline{\phi}, \underline{\phi} \gets$ {\tt IR}($G$)\;
    push $V'$ into $stk$\;
    \While{$k>0\wedge stk\neq \emptyset$} {
        $S\gets$ pop out the first subset from $stk$\;
        \tcp{\textcolor{teal}{(1) The vertex weight updating method.}}
        $\bm{w}\gets${\tt VWU}($G$,$\bm{w}$, $\overline{\phi}, \underline{\phi}$)\;
        \tcp{\textcolor{teal}{(2) The graph reduction and division method.}}
        $S, stk, \overline{\phi}, \underline{\phi}\gets${\tt GRD}($G$, $\bm{w}$, $stk$)\;
        \tcp{\textcolor{teal}{(3) The candidate subgraph extraction and verification method.}}
        $k\gets${\tt CSEV}($G$, $S$, $k$)\;
    }
\end{algorithm}

\Cref{alg:framework} depicts this framework, including four components: {\it Initial Reduction} ({\tt IR}), {\it Vertex Weight Updating} ({\tt VWU}), {\it Graph Reduction and Division} ({\tt GRD}), and {\it Candidate Subgraph Extraction and Verification} ({\tt CSEV}). Specifically, given a graph $G$ and an integer $k$, we initialize the upper and lower bounds of compact (resp. tr-compact) numbers using core (resp. tr-core) numbers (line 2). We then iteratively execute operations in the following three stages (lines 4-8):
\begin{enumerate}
    \item Update the vertex weight vector $\bm{w}$ for each vertex;

    \item Prune invalid vertices via upper and lower bounds of $\phi$, update bounds of $\phi$ in residual graph, and divide the large graph into multiple smaller subgraphs for further computation;

    \item {Extract the candidate subgraph using vertex weight vector $\bm{w}$, and verify if the candidate subgraph meets the requirements.}
\end{enumerate}

\begin{table*}[t]
    \centering
    \small
    \caption{{Overview of the four components of LDS/LTDS solutions.}}
    \begin{tabular}{l|l|p{2.4cm}|l|l}
        \toprule
        \textbf{Method}                                                     & \textbf{IR}
                                                                            & \textbf{Stage (1):} {\tt VWU}
                                                                            & \textbf{Stage (2):} {\tt GRD}
                                                                            & \textbf{Stage (3):} {\tt CSEV}
        \\ \hline \hline
        \multirow{3}{*}{{\tt LDSflow}}                                      & \multirow{3}{*}{$k$-core}                                         & \multirow{6}{*}{\parbox{2.4cm}{compute the maximum flow}} & {\small\ding{182}}  prune via \Cref{cor:uv-prune} and \Cref{lem:uu-prune};                  & \multirow{3}{*}{{\small\ding{182}} extract the minimum cut; }         \\
                                                                            &                                                                   &                                                           & {\small\ding{183}} no updating;                                                             &                                                                       \\
                                                                            &                                                                   &                                                           & {\small\ding{184}} extract the residual graph.                                              &                                                                       \\\cline{1-2} \cline{4-4}
        \multirow{3}{*}{{\tt LTDSflow}}                                     & \multirow{3}{*}{tr-$k$-core}                                      &                                                           & {\small\ding{182}}  no pruning;                                                             & \multirow{3}{*}{{\small\ding{183}} verify locally densest property. } \\
                                                                            &                                                                   &                                                           & {\small\ding{183}} no updating;                                                             &                                                                       \\
                                                                            &                                                                   &                                                           & {\small\ding{184}} extract the residual graph.                                              &                                                                       \\\hline
        \multirow{3}{*}{{\tt LDScvx}}                                       & \multirow{3}{*}{$k$-core}                                         & \multirow{6}{*}{\parbox{2.4cm}{optimize CP formulation}}  & \multirow{2}{*}{{\small\ding{182}}  prune via \Cref{cor:uv-prune} and \Cref{lem:uu-prune};} & \multirow{2}{*}{{\small\ding{182}} extract the top stable group; }    \\
                                                                            &                                                                   &                                                           &                                                                                             &                                                                       \\
                                                                            &                                                                   &                                                           &
        \multirow{2}{*}{{\small\ding{183}} update bounds via stable group;} & \multirow{2}{*}{{\small\ding{183}} verify self-densest property;}                                                                                                                                                                                                                                   \\\cline{1-2}
        \multirow{3}{*}{{\tt LTDScvx}}                                      & \multirow{3}{*}{tr-$k$-core}                                      &                                                           &                                                                                             &                                                                       \\
                                                                            &                                                                   &                                                           & \multirow{2}{*}{{\small\ding{184}} divide the graph via stable group.}                      & \multirow{2}{*}{{\small\ding{184}} verify locally densest property.}  \\
                                                                            &                                                                   &                                                           &                                                                                             &                                                                       \\\bottomrule
    \end{tabular}

    \label{tab:brief_overview}
\end{table*}

We outline the four components of different algorithms in \Cref{tab:brief_overview}. For example, for our {\tt LDScvx}, it uses $k$-core to initialize bounds of $\phi$. Then, for the repeated searching process: {\bf Stage (1)} employ \FW to optimize \Cref{equ:cp}; {\bf Stage (2)} prune invalid vertices based on \Cref{cor:uv-prune} and \Cref{lem:uu-prune}, while adopting stable group to tighten the bounds and decompose the problem into multiple sub-problems; {\bf Stage (3)} extract the top stable group as a candidate, and execute \IsDensest and \IsLDS to verify its optimality.

    {\bf Comparison between the max flow-based and CP-based algorithms.}
With the help of our unified framework, we can thoroughly compare our method and the max flow-based ones.
For the {\tt IR} component, they share the same initialization. For {\tt VWU} stage, the vertex weight vector $\bm{w}$ represents the weights assigned to different vertices. It holds different meanings and serves different purposes in different algorithms:

{\begin{itemize}
    \item In the max-flow-based algorithms, $\bm{w}$ denotes the flows from the vertices to the target node;
    \item In our CP-based algorithms, $\bm{w}$ represents the weight sum received by each vertex.
\end{itemize}
This stage reflects the main difference between these two types of algorithms. For {\tt GRD} stage, with the help of stable groups, our CP-based algorithms provide {\em tighter} bounds via convex programming, which enables more vertices to be pruned and allows a division of the original graph. Furthermore, in the {\tt CSEV} stage, {\tt LDScvx} and {\tt LTDScvx} construct {\em smaller} flow networks to verify the optimality of the candidate by exploiting tighter bounds.

    {\bf Comparison between our framework and the DSD framework \cite{zhou2024depth}.} First, we divide the {\it Graph Reduction} stage into two parts, {\tt IR} and {\tt GRD}, since these two parts serve different purposes and utilize different notions, such as $k$-core and stable groups, in the LDS problem and its variant. In contrast, for the DSD problem, $k$-core is the only concept used in this stage. Second, since the LDS and LTDS problems aim to identify multiple dense regions in the given graph, both reduction and division are necessary. However, the DSD problem focuses solely on the densest part, making reduction the only essential step.
}

\color{black}
\section{Experiments}
\label{sec:exp}

\subsection{Setup}
\label{sec:exp:setup}

We use thirteen real datasets \cite{snapnets,nr,yang2015defining,boldi2011layered,boldi2004webgraph,nr} which are publicly available\footnote{\url{https://networkrepository.com/}, \url{https://snap.stanford.edu/data/index.html}, and \url{https://law.di.unimi.it/datasets.php}} except for TL. The TL dataset is provided by a television company TCL Technology. The dataset contains film information provided on its smart TV platform, mainly used for a case study. Other graph datasets cover various domains, including social networks (e.g., LiveJournal), e-commerce (e.g., Amazon), and video platforms (e.g., YouTube). \Cref{tab:datasets} summarizes the statistics.



\begin{table}[t]
    \centering
    \caption{Graphs used in our experiments.}
    \algofont
    \begin{tabular}{l|l|r|r}
        \hline
        Dataset & Category      & $|V|$ & $|E|$ \\
        \hline\hline
        PG      & Social        & 10.7K & 24.3K \\ \hline
        BK      & Social        & 58.2K & 214K  \\ \hline
        TL      & Movie         & 108K  & 168K  \\ \hline
        GW      & Social        & 196K  & 950K  \\ \hline
        AM      & E-commerce    & 335K  & 926K  \\ \hline
        YT      & Video-sharing & 1.13M & 2.99M \\ \hline
        SK      & Communication & 1.70M & 11.1M \\ \hline
        LJ      & Social        & 4.00M & 34.7M \\ \hline
        OR      & Social        & 3.07M & 117M  \\ \hline
        \MD{IC} & Web           & 7.41M & 194M  \\ \hline
        \MD{AB} & Web           & 22.7M & 639M  \\ \hline
        IT      & Web           & 41.3M & 1.03B \\ \hline
        LK      & Hyperlink     & 52.6M & 1.61B \\
        \hline
    \end{tabular}%
    \label{tab:datasets}%
\end{table}

We compare the following algorithms:
\begin{itemize}
    \item {\tt LDScvx} is our convex-programming based top-$k$ LDS algorithm (\Cref{sec:algo:lds}).
    \item {\tt LDSflow} \cite{qin2015locally} is the state-of-the-art top-$k$ LDS algorithm based on max-flow.
    \item {\tt LTDScvx} is our convex-programming based top-$k$ LTDS algorithm (\Cref{sec:ltds}).
    \item {\tt LTDSflow} \cite{samusevich2016local} is the state-of-the-art top-$k$ LTDS algorithm based on max-flow.
\end{itemize}

All the algorithms above are implemented in C++ with STL used, except for {\tt LDSflow}, whose source codes are provided by the authors of \cite{qin2015locally}. We run all the experiments on a machine having an Intel(R) Xeon(R) Silver 4110 CPU @ 2.10GHz processor and 256GB memory, with Ubuntu installed.

\subsection{Overall Evaluation of LDS and LTDS Algorithms}
\textbf{Setting of $N$.}\MD{To choose the best setting of $N$, i.e., the number of \FW iterations, we tested the running time of {\tt LDScvx} w.r.t. different values of $N$ from 50 to 200 with $k$ fixed to $5$. \Cref{tab:para-N} reports the average relative running time w.r.t. $N$ over different datasets. The relative running time for a specific value of $N$ on each dataset is obtained via dividing the running time by the minimum running time over all $N$ values for the dataset. The mean is then obtained by averaging over all datasets. We can find that when $N=100$ (resp. $N=200$), we obtain the minimum average relative running time. Hence we use $N=100$ (resp. $N=200$) as the default parameter value in other experiments.}

\begin{table}[t]
    \centering
    \caption{\MD{Relative running time w.r.t. different $N$}}
    \begin{tabular}{r|r|r|r|r}
        \hline
        $N$           & 50   & 100  & 150  & 200  \\
        \hline
        \hline
        {\tt LDScvx}  & 1.64 & 1.10 & 1.12 & 1.20 \\
        \hline
        {\tt LTDScvx} & 1.54 & 1.42 & 1.27 & 1.24 \\
        \hline
    \end{tabular}%
    \label{tab:para-N}%
\end{table}%

\noindent\textbf{Evaluation of efficiency.} In this experiment, we compare our top-$k$ LDS algorithm {\tt LDScvx} with the state-of-the-art {\tt LDSflow} \cite{qin2015locally} w.r.t. running time.

We first fix $k=5$ to overview the two algorithms on different datasets. For the edge-based setting, we conduct experiments on the nine largest datasets. \yy{In contrast, for the triangle-based setting, due to its higher computational complexity, we test on six smaller datasets, excluding TL, as it contains no triangles.}
\Cref{fig:eff-overall} shows the efficiency results of the two algorithms. The datasets are ordered by graph size on the x-axis. Note that for some datasets, the bars of {\tt LDSflow} touch the solid upper line, which means {\tt LDSflow} cannot finish within 600 hours on those datasets. From \Cref{fig:eff-overall}, we can observe: first, the running time of {\tt LDSflow} increases along with the graph size increasing; second, {\tt LDScvx} is up to four orders of magnitude faster than {\tt LDSflow}.
\yy{For the LTDS algorithms, a similar trend can be observed: {\tt LTDSflow} generally requires significantly more time than {\tt LTDScvx}, and {\tt LTDScvx} is up to more than 20 times faster than {\tt LTDSflow} on certain datasets.}
We reckon that the speedup comes from \MD{more vertices pruned due to tighter bounds, fewer candidates verified, and smaller flow networks as stated in \Cref{sec:algo}.}



\begin{figure*}[t]
    \centering
    \begin{minipage}{0.48\textwidth}
        \centering
        \input{figs/time-overall}
        \vspace{-1em}
        \caption{Efficiency of {\tt LDScvx} and {\tt LDSflow} with $k=5$.}
        \label{fig:eff-overall}
    \end{minipage}%
    \hfill
    \begin{minipage}{0.48\textwidth}
        \centering
        \input{figs/time-overall-ltds}
        \vspace{-1em}
        {\color{blackkkk} \caption{{Efficiency of {\tt LTDScvx} and {\tt LTDSflow} with $k=5$.}}}
        \label{fig:eff-overall-ltds}
    \end{minipage}
\end{figure*}

\begin{figure*}[t]
    \centering
    \begin{minipage}{0.48\textwidth}
        \centering
        \input{figs/time-k}
        \vspace{-1em}
        \caption{Efficiency of {\tt LDScvx} and {\tt LDSflow} w.r.t. different $k$.}
        \label{fig:eff-k}
    \end{minipage}%
    \hfill
    \begin{minipage}{0.48\textwidth}
        \centering
        \input{figs/time-k-ltds}
        \vspace{-1em}
        {\color{blackkkk} \caption{Efficiency of {\tt LTDScvx} and {\tt LTDSflow} w.r.t. different $k$.}}
        \label{fig:eff-k-ltds}
    \end{minipage}
\end{figure*}


We further provide the running time trends of these algorithms w.r.t. different $k$ values on \MD{four representative datasets in \Cref{fig:eff-k} and \Cref{fig:eff-k-ltds} due to space limit}. For {\tt LDSflow}, we do not show its trends on dataset LK because it cannot finish within 600 hours even with $k=5$. We can see that when $k$ increases, the running time of all algorithms increases. Moreover, the growth rates of {\tt LDScvx}’s and {\tt LTDScvx}'s running times with respect to $k$ is generally smaller than those of {\tt LDSflow} and {\tt LTDSflow}.

\noindent\textbf{Memory usage.} Further, we test the memory usage of these algorithms. The maximum memory usage is tested via the Linux command {\tt /usr/bin/time -v}. For the cases that {\tt LDSflow} does not finish reasonably, we record the maximum resident memory during the running process.
\Cref{fig:memory} and \Cref{fig:memory-ltds} report the maximum memory usage of these algorithms. The datasets are sorted on the x-axis in ascending order of graph size. We can observe that the memory usages of all algorithms increase along with the increasing graph size. Besides, the memory costs of convex-programming-based algorithms and max-flow-based algorithms are around the same scale because all algorithms take linear memory usage w.r.t. the graph size.


\begin{figure*}[t]
    \centering
    \begin{minipage}{0.48\textwidth}
        \centering
        \input{figs/memory-usage}
        \vspace{-1em}
        \caption{Memory usage of LDS algorithms with $k=5$.}
        \label{fig:memory}
    \end{minipage}%
    \hfill
    \begin{minipage}{0.48\textwidth}
        \centering
        \input{figs/memory-usage-ltds}
        \vspace{-1em}
        {\color{blackkkk} \caption{Memory usage of LTDS algorithms with $k=5$.}
        \label{fig:memory-ltds}}
    \end{minipage}
\end{figure*}

\subsection{Comparison Under the Unified Framework}
In this experiment, we comprehensively compare all algorithms under our summarized unified framework for LDS and LTDS problems. Since the {\tt IR} methods are nearly identical across all algorithms, and the {\tt VWU} methods are difficult to compare in isolation, we omit the discussion of these two components.


\color{blackkkk}
\noindent\textbf{Evaluation of {\tt GRD} methods.} The primary difference between the {\tt GRD} methods used in our convex-programming-based algorithms and those in existing max-flow-based solutions lies in the way {\tt LDScvx} and {\tt LTDScvx} derive stable groups to provide tighter bounds, thereby enabling more effective vertex pruning and a more efficient decomposition of the original graph.

To analyze the priority of stable groups, we focus on two special cases shown in \Cref{fig:eff-k} and \Cref{fig:eff-k-ltds}:
\begin{itemize}
    \item In \Cref{fig:eff-k}, there is a significant increase in {\tt LDSflow} (about two orders of magnitude) when $k$ is increased from $10$ to $15$ on dataset YT;
    \item In \Cref{fig:eff-k-ltds},  the running time of {\tt LTDScvx} on dataset AM increases more noticeably with $k$ compared to other datasets.
\end{itemize}
Here, we use the number of failed LDS and LTDS candidates as the evaluation metric for the {\tt GRD} method, since better pruning and problem decomposition would lead to fewer but more accurate verifications.
\color{black}
\begin{table}[htbp]
    \centering
    \caption{Numbers of failed LDS candidates on YT w.r.t. $k$.}
    \begin{tabular}{l|r|r|r}
        \hline
        Algorithm     & $k=10$ & $k=15$ & increased times \\
        \hline
        \hline
        {\tt LDScvx}  & 37     & 84     & 2.27$\times$    \\
        \hline
        {\tt LDSflow} & 277    & 18399  & 66.42$\times$   \\
        \hline
    \end{tabular}%
    \label{tab:failed-YT}%
\end{table}%

\begin{table}[htbp]
    \centering
    {\color{blackkkk} \caption{Numbers of failed LTDS candidates on AM w.r.t. $k$.}}
    \begin{tabular}{l|r|r|r}
        \hline
        Algorithm     & $k=10$ & $k=15$ & increased times \\
        \hline
        \hline
        {\tt LDScvx}  & 0      & 0      & ---             \\
        \hline
        {\tt LDSflow} & 0      & 0      & ---             \\
        \hline
    \end{tabular}%
    \label{tab:failed-YT}%
\end{table}%

We report the numbers of failed LDS candidates on YT with $k=10$ and $15$ for both algorithms in \Cref{tab:failed-YT}. We observe that the number of failed candidates in {\tt LDSflow} increased around 66$\times$ when $k$ increases from $10$ to $15$, which explains the surge of running time in \Cref{fig:eff-k}. In contrast, the number of failed candidates in {\tt LDScvx} only increases by about 2$\times$. This is why the running time of {\tt LDScvx} does not increase much when $k$ is increased from 10 to 15. Another observation from \Cref{tab:failed-YT} is that the failed numbers for {\tt LDScvx} are smaller than {\tt LDSflow} on both $k$ values, respectively. The reason is that we provide tight upper and lower bounds for compact numbers via convex programming and stable groups, and the tight bounds further enable more vertices to be pruned, which results in fewer LDS candidates to be examined. Apart from that, to examine an LDS candidate, our {\tt LDScvx} only needs a subgraph of what is needed in {\tt LDSflow} to calculate the max-flow by leveraging the lower bounds of compact numbers, according to \Cref{sec:algo:verify}. We believe the above two improvements explain why we are around three orders of magnitude faster on YT when $k=15$.
\color{blackkkk}
We further analyze the special case observed on the AM dataset in \Cref{fig:eff-k-ltds}, where the running time of {\tt LTDScvx} increases more noticeably with $k$ compared to other datasets. Specifically, we find that when $k=25$, the number of failed LTDS candidates for both {\tt LTDScvx} and {\tt LTDSflow} drops to zero. This indicates that the top-25 locally triangle-dense regions in AM can already be effectively pruned and decomposed using the tr-$k$-core number, without requiring the additional guidance from stable groups. As a result, the tighter bounds provided by our {\tt GRD} method do not further contribute to vertex pruning or problem division in this case. This explains why {\tt LTDScvx} does not exhibit a clear efficiency advantage and its running time increases with $k$, similar to {\tt LTDSflow}.

These two special cases demonstrate that the efficiency of {\tt LDScvx} and {\tt LTDScvx} heavily depend on the effectiveness of the {\tt GRD} method, particularly in providing tighter bounds for pruning and division. When the bounds are already tight (e.g., due to strong core structure), the advantage of stable groups diminishes, leading to a higher growth rate as $k$ increases.
\color{black}

\noindent\textbf{Evaluation of {\tt CSEV} methods.} Our main optimization for the {\tt CSEV} method is that we construct smaller flow networks to verify the optimality of candidates. Here, we conduct an ablation study on {\tt LDScvx} and {\tt LTDScvx} to understand the effectiveness of {\tt IsLDS} (\Cref{alg:is-lds}) and \IsLTDS (\Cref{alg:is-ltds}).

Recall that in {\tt IsLDS} we only include the vertices satisfying $\overline{\phi}(u) \geq \mathsf{density}(G[S]) \geq \underline{\phi}(u)$ into the flow network computation, while its counterpart in {\tt LDSflow} \cite{qin2015locally} includes all vertices satisfying $\overline{\phi}(u) \geq \mathsf{density}(G[S])$, named by {\tt IsLDS-core}. We report the time used for verifying LDSs with {\tt IsLDS} and {\tt IsLDS-core}, respectively, on the nine datasets when $k=5$ in \Cref{tab:ablation}. \MD{Here, the time of {\tt IsLDS-core} is measured by replacing {\tt IsLDS} with {\tt IsLDS-core} in {\tt LDScvx}.} The time of {\tt IsLDS-core} on LK is marked as $\geq 259200$s because it cannot finish within three days. From the results, we observe that the verification process with {\tt IsLDS} is up to 110$\times$ faster than that with {\tt IsLDS-core}.

\color{blackkkk}
For the LTDS problem, we compare the following three verification strategies:
\begin{itemize}
    \item {\tt IsLTDS}: our proposed method, which includes only the vertices satisfying $\overline{\phi}_{tr}(u) \geq \mathsf{tr\text{-}density}(G[S]) \geq \underline{\phi}_{tr}(u)$ in the flow network construction;
    \item {\tt IsLTDS-core}: a generalization of {\tt IsLDS-core} to the triangle-based setting, including all vertices satisfying $\overline{\phi}_{tr}(u) \geq \mathsf{tr\text{-}density}(G[S])$;
    \item {\tt IsLTDS-bs} (baseline)\cite{samusevich2016local}: the method adopted in {\tt LTDSflow}, which includes all vertices in $G$ for verification.
\end{itemize}

We report the running time of the three verification methods on six datasets with $k=5$ in \Cref{fig:ablation}. As shown in the figure, {\tt IsLTDS} consistently achieves the lowest verification cost across all datasets. In particular, on large graphs such as SK and YT, {\tt IsLTDS} reduces the verification time by up to two orders of magnitude compared to {\tt IsLTDS-bs}, and is significantly faster than {\tt IsLTDS-core} as well. This improvement is attributed to the tight upper and lower bounds derived from the stable group, which help exclude many irrelevant vertices from the flow network construction. As a result, {\tt IsLTDS} performs verification over a much smaller subgraph, leading to substantial efficiency gains. Moreover, compared to {\tt IsLTDS-core}, which only applies the upper bound constraint, {\tt IsLTDS} leverages both bounds to further reduce the search space, especially when the bounds are tight. The comparison on AM further highlights this: while {\tt IsLTDS} is significantly faster than the other two methods on most datasets, it shows little advantage over {\tt IsLTDS-core} on AM. This observation is consistent with the special case discussed earlier, where we showed that all top-$k$ candidates on AM can already be effectively pruned using the tr-$k$-core number alone. In such cases, the additional bounds provided by the stable group in {\tt IsLTDS} do not further reduce the search space, making its verification cost comparable to that of {\tt IsLTDS-core}.
\color{black}


\begin{table}[t]
    \centering
    {\caption{Effect of {\tt IsLDS} with $k=5$.}}
    \begin{tabular}{l|r|r|r}
        \hline
        Dataset & {\tt IsLDS} & {\tt IsLDS-core} & Speedup             \\
        \hline
        \hline
        AM      & 0.3334s     & 0.3623s          & $1.09\times$        \\
        \hline
        YT      & 2.6575s     & 80.9994s         & $30.48\times$       \\
        \hline
        SK      & 58.7789s    & 18035.1864s      & $306.83\times$      \\
        \hline
        LJ      & 2.1204s     & 2.3924s          & $1.13\times$        \\
        \hline
        OR      & 18.4089s    & 723.6035s        & $39.31\times$       \\
        \hline
        \MD{IC} & 285.4502s   & 288.9184s        & $1.01\times$        \\
        \hline
        \MD{AB} & 60.2669s    & 62.0416s         & $1.03\times$        \\
        \hline
        IT      & 147.9361s   & 188.8527s        & $1.28 \times$       \\
        \hline
        LK      & 2335.4461s  & $\geq 259200$s   & $\geq 110.99\times$ \\
        \hline
    \end{tabular}%
    \label{tab:ablation}%
\end{table}%

\begin{figure}[t]
    \input{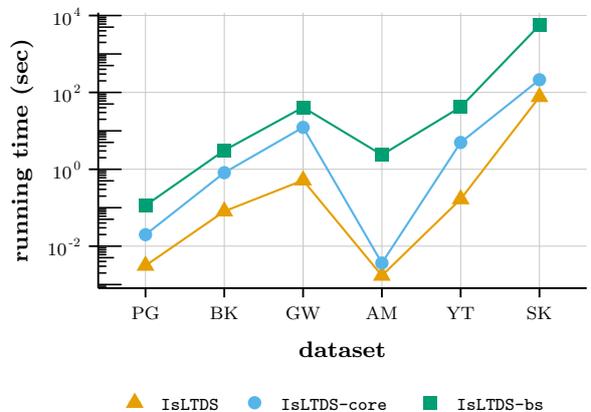}
    \vspace{-1em}
    {\color{blackkkk} \caption{Effect of {\tt IsLTDS} with $k=5$.}}
    \label{fig:ablation}
\end{figure}

\subsection{Time Proportion Analysis}

\begin{figure*}[t]
    \centering
    \begin{minipage}{0.48\textwidth}
        \centering
        \input{figs/time-proportion}
        \vspace{-1em}
        \caption{Proportion of each part in total running time for {\tt LDScvx}.}
        \label{fig:time-proportion}
    \end{minipage}%
    \hfill
    \begin{minipage}{0.48\textwidth}
        \centering
        \input{figs/time-proportion-ltds}
        \vspace{-1em}
        \caption{Proportion of each part in total running time for {\tt LTDScvx}.}
        \label{fig:time-proportion-ltds}
    \end{minipage}
\end{figure*}


Here, we evaluate how the different building blocks of {\tt LDScvx} and {\tt LTDScvx} contribute to the whole running time after the graph is loaded and preprocessed. \Cref{fig:time-proportion} reports the proportion of each part in the total running time for {\tt LDScvx}: Frank-Wolfe (\Cref{alg:fw}), ExtractSG (\Cref{alg:ext-sg}), Pruning (\Cref{alg:pruning}), and VerifyLDS (\Cref{alg:is-lds}) with $k=5$ over nine datasets. We can observe that the Frank-Wolfe computation is the most computationally expensive part on most datasets. For SK, LK and YT datasets, the time used by verifying LDS takes the majority. We further examine the number of failed LDS candidates (i.e., on which \IsLDS returns {\sf False}) on the nine datasets. \Cref{tab:failed-can} reports the results. We can find that the numbers of failed candidates on these datasets are much higher than other datasets, which means that these two datasets need more time to verify LDSs, which explains the results in \Cref{fig:time-proportion} to some extent. \MD{For IC, the time used by verifying LDS is also high, because the first LDS in IC is quite large and takes relatively long time to verify.}
\color{blackkkk}
For {\tt LTDScvx}, the overall time breakdown in \Cref{fig:time-proportion-ltds} is similar to that of {\tt LDScvx}, with Frank-Wolfe and {\tt VerifyLTDS} being the most time-consuming components across most datasets. The correlation between verification time and the number of failed candidates is further supported by \Cref{tab:failed-can-ltds}, which shows that datasets with more failed candidates also exhibit higher {\tt VerifyLTDS} proportions.
\color{black}
\begin{table}[t]
    \centering
    \caption{Numbers of failed LDS candidates with $k=5$.}
    \algofont
    \begin{tabular}{l|r|r|r|r|r|r|r|r|r}
        \hline
        Dataset  & AM & YT & SK & LJ & OR & \MD{IC} & \MD{AB} & IT & LK \\
        \hline
        \hline
        \#failed & 0  & 9  & 67 & 1  & 1  & 0       & 0       & 0  & 6  \\
        \hline
    \end{tabular}%
    \label{tab:failed-can}%
\end{table}%

\begin{table}[t]
    \centering
    {\color{blackkkk} \caption{Numbers of failed LTDS candidates with $k=5$.}}
    \algofont
    \begin{tabular}{l|r|r|r|r|r|r}
        \hline
        Dataset  & PG & BK & GW & AM & YT & SK \\
        \hline
        \hline
        \#failed & 0  & 8  & 18 & 0  & 7  & 21 \\
        \hline
    \end{tabular}%
    \label{tab:failed-can-ltds}%
\end{table}%

\subsection{Scalability} \MD{We further tested the scalability w.r.t. the density via synthetics datasets. The six synthetic graphs are generated by Barab\'{a}si-Albert (BA) model \cite{albert2002statistical}. The numbers of vertices in all synthetic graphs are fixed to 1,000,000, and the densities are increased linearly from 2 to 12. In other words, the numbers of edges are from 2,000,000 to 12,000,000. We report the running time of {\tt LDScvx} on the six synthetic datasets in \Cref{fig:scalability}. We observe that {\tt LDScvx} scales well w.r.t. the graph density.}

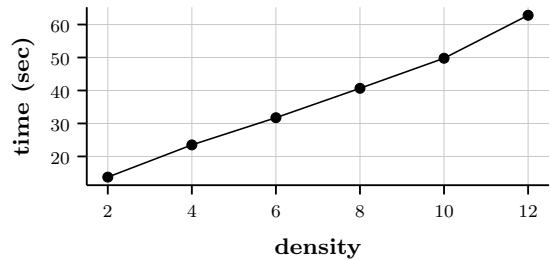
\begin{figure}[h]
\begin{tikzpicture}[x=1pt,y=1pt]
\definecolor{fillColor}{RGB}{255,255,255}
\path[use as bounding box,fill=fillColor,fill opacity=0.00] (0,0) rectangle (216.81,115.63);
\begin{scope}
\path[clip] (  0.00,  0.00) rectangle (216.81,115.63);
\definecolor{fillColor}{RGB}{255,255,255}

\path[fill=fillColor] ( -0.00,  0.00) rectangle (216.81,115.63);
\end{scope}
\begin{scope}
\path[clip] ( 35.18, 34.25) rectangle (208.27,101.41);
\definecolor{fillColor}{RGB}{255,255,255}

\path[fill=fillColor] ( 35.18, 34.25) rectangle (208.27,101.41);

\path[] ( 35.18, 38.88) --
	(208.27, 38.88);

\path[] ( 35.18, 51.32) --
	(208.27, 51.32);

\path[] ( 35.18, 63.76) --
	(208.27, 63.76);

\path[] ( 35.18, 76.20) --
	(208.27, 76.20);

\path[] ( 35.18, 88.64) --
	(208.27, 88.64);

\path[] ( 35.18,101.08) --
	(208.27,101.08);

\path[] ( 58.78, 34.25) --
	( 58.78,101.41);

\path[] ( 90.25, 34.25) --
	( 90.25,101.41);

\path[] (121.73, 34.25) --
	(121.73,101.41);

\path[] (153.20, 34.25) --
	(153.20,101.41);

\path[] (184.67, 34.25) --
	(184.67,101.41);
\definecolor{drawColor}{RGB}{199,200,199}

\path[draw=drawColor,line width= 0.2pt,line join=round] ( 35.18, 45.10) --
	(208.27, 45.10);

\path[draw=drawColor,line width= 0.2pt,line join=round] ( 35.18, 57.54) --
	(208.27, 57.54);

\path[draw=drawColor,line width= 0.2pt,line join=round] ( 35.18, 69.98) --
	(208.27, 69.98);

\path[draw=drawColor,line width= 0.2pt,line join=round] ( 35.18, 82.42) --
	(208.27, 82.42);

\path[draw=drawColor,line width= 0.2pt,line join=round] ( 35.18, 94.86) --
	(208.27, 94.86);

\path[draw=drawColor,line width= 0.2pt,line join=round] ( 43.05, 34.25) --
	( 43.05,101.41);

\path[draw=drawColor,line width= 0.2pt,line join=round] ( 74.52, 34.25) --
	( 74.52,101.41);

\path[draw=drawColor,line width= 0.2pt,line join=round] (105.99, 34.25) --
	(105.99,101.41);

\path[draw=drawColor,line width= 0.2pt,line join=round] (137.46, 34.25) --
	(137.46,101.41);

\path[draw=drawColor,line width= 0.2pt,line join=round] (168.93, 34.25) --
	(168.93,101.41);

\path[draw=drawColor,line width= 0.2pt,line join=round] (200.41, 34.25) --
	(200.41,101.41);
\definecolor{drawColor}{RGB}{0,0,0}

\path[draw=drawColor,line width= 0.6pt,line join=round] ( 43.05, 37.30) --
	( 74.52, 49.45) --
	(105.99, 59.71) --
	(137.46, 70.80) --
	(168.93, 82.14) --
	(200.41, 98.35);
\definecolor{fillColor}{RGB}{0,0,0}

\path[draw=drawColor,line width= 0.4pt,line join=round,line cap=round,fill=fillColor] ( 43.05, 37.30) circle (  1.86);

\path[draw=drawColor,line width= 0.4pt,line join=round,line cap=round,fill=fillColor] ( 74.52, 49.45) circle (  1.86);

\path[draw=drawColor,line width= 0.4pt,line join=round,line cap=round,fill=fillColor] (105.99, 59.71) circle (  1.86);

\path[draw=drawColor,line width= 0.4pt,line join=round,line cap=round,fill=fillColor] (137.46, 70.80) circle (  1.86);

\path[draw=drawColor,line width= 0.4pt,line join=round,line cap=round,fill=fillColor] (168.93, 82.14) circle (  1.86);

\path[draw=drawColor,line width= 0.4pt,line join=round,line cap=round,fill=fillColor] (200.41, 98.35) circle (  1.86);

\path[] ( 35.18, 34.25) rectangle (208.27,101.41);
\end{scope}
\begin{scope}
\path[clip] (  0.00,  0.00) rectangle (216.81,115.63);
\definecolor{drawColor}{RGB}{0,0,0}

\path[draw=drawColor,line width= 0.7pt,line join=round] ( 35.18, 34.25) --
	( 35.18,101.41);
\end{scope}
\begin{scope}
\path[clip] (  0.00,  0.00) rectangle (216.81,115.63);
\definecolor{drawColor}{RGB}{0,0,0}

\node[text=drawColor,anchor=base east,inner sep=0pt, outer sep=0pt, scale=  0.80] at ( 28.88, 42.34) {20};

\node[text=drawColor,anchor=base east,inner sep=0pt, outer sep=0pt, scale=  0.80] at ( 28.88, 54.79) {30};

\node[text=drawColor,anchor=base east,inner sep=0pt, outer sep=0pt, scale=  0.80] at ( 28.88, 67.23) {40};

\node[text=drawColor,anchor=base east,inner sep=0pt, outer sep=0pt, scale=  0.80] at ( 28.88, 79.67) {50};

\node[text=drawColor,anchor=base east,inner sep=0pt, outer sep=0pt, scale=  0.80] at ( 28.88, 92.11) {60};
\end{scope}
\begin{scope}
\path[clip] (  0.00,  0.00) rectangle (216.81,115.63);
\definecolor{drawColor}{RGB}{0,0,0}

\path[draw=drawColor,line width= 0.7pt,line join=round] ( 31.68, 45.10) --
	( 35.18, 45.10);

\path[draw=drawColor,line width= 0.7pt,line join=round] ( 31.68, 57.54) --
	( 35.18, 57.54);

\path[draw=drawColor,line width= 0.7pt,line join=round] ( 31.68, 69.98) --
	( 35.18, 69.98);

\path[draw=drawColor,line width= 0.7pt,line join=round] ( 31.68, 82.42) --
	( 35.18, 82.42);

\path[draw=drawColor,line width= 0.7pt,line join=round] ( 31.68, 94.86) --
	( 35.18, 94.86);
\end{scope}
\begin{scope}
\path[clip] (  0.00,  0.00) rectangle (216.81,115.63);
\definecolor{drawColor}{RGB}{0,0,0}

\path[draw=drawColor,line width= 0.7pt,line join=round] ( 35.18, 34.25) --
	(208.27, 34.25);
\end{scope}
\begin{scope}
\path[clip] (  0.00,  0.00) rectangle (216.81,115.63);
\definecolor{drawColor}{RGB}{0,0,0}

\path[draw=drawColor,line width= 0.7pt,line join=round] ( 43.05, 30.75) --
	( 43.05, 34.25);

\path[draw=drawColor,line width= 0.7pt,line join=round] ( 74.52, 30.75) --
	( 74.52, 34.25);

\path[draw=drawColor,line width= 0.7pt,line join=round] (105.99, 30.75) --
	(105.99, 34.25);

\path[draw=drawColor,line width= 0.7pt,line join=round] (137.46, 30.75) --
	(137.46, 34.25);

\path[draw=drawColor,line width= 0.7pt,line join=round] (168.93, 30.75) --
	(168.93, 34.25);

\path[draw=drawColor,line width= 0.7pt,line join=round] (200.41, 30.75) --
	(200.41, 34.25);
\end{scope}
\begin{scope}
\path[clip] (  0.00,  0.00) rectangle (216.81,115.63);
\definecolor{drawColor}{RGB}{0,0,0}

\node[text=drawColor,anchor=base,inner sep=0pt, outer sep=0pt, scale=  0.80] at ( 43.05, 22.44) {2};

\node[text=drawColor,anchor=base,inner sep=0pt, outer sep=0pt, scale=  0.80] at ( 74.52, 22.44) {4};

\node[text=drawColor,anchor=base,inner sep=0pt, outer sep=0pt, scale=  0.80] at (105.99, 22.44) {6};

\node[text=drawColor,anchor=base,inner sep=0pt, outer sep=0pt, scale=  0.80] at (137.46, 22.44) {8};

\node[text=drawColor,anchor=base,inner sep=0pt, outer sep=0pt, scale=  0.80] at (168.93, 22.44) {10};

\node[text=drawColor,anchor=base,inner sep=0pt, outer sep=0pt, scale=  0.80] at (200.41, 22.44) {12};
\end{scope}
\begin{scope}
\path[clip] (  0.00,  0.00) rectangle (216.81,115.63);
\definecolor{drawColor}{RGB}{0,0,0}

\node[text=drawColor,anchor=base,inner sep=0pt, outer sep=0pt, scale=  1.00] at (121.73,  8.15) {\bfseries density};
\end{scope}
\begin{scope}
\path[clip] (  0.00,  0.00) rectangle (216.81,115.63);
\definecolor{drawColor}{RGB}{0,0,0}

\node[text=drawColor,rotate= 90.00,anchor=base,inner sep=0pt, outer sep=0pt, scale=  1.00] at ( 13.49, 67.83) {\bfseries time (sec)};
\end{scope}
\end{tikzpicture}
    \vspace{-1em}
    \caption{\MD{The scalability w.r.t. the density.}}
    \label{fig:scalability}
\end{figure}

\subsection{Subgraph Statistics}
\label{sec:exp:lds-stat}
\MD{\Cref{fig:density-size} reports the densities of the top-15 densest subgraphs w.r.t. the size (number of vertices) returned by four different models on four datasets, where {\tt Greedy} iteratively computes a densest subgraph and removes it from the graph, and {\tt FDS} denotes the density-friendly decomposition model \cite{tatti2015density,danisch2017large}. From \Cref{fig:density-size}, we can find that the densest subgraph can be found by all three algorithms, except {\tt LTDS}, because the densest subgraph is also an LDS. But there are some  dense subgraphs found by {\tt Greedy} that do not qualify as LDSs. Hence, the subgraphs found by {\tt LDScvx} have a wide range of densities and sizes. For {\tt FDS}, we can find that the subgraphs have increasing sizes and decreasing densities. This is because {\tt FDS} outputs a chain of subgraphs, where each subgraph is nested within the next one, and the inner one is denser than the outer ones. \textcolor{blackkkk}{In contrast, {\tt LTDS} operates under a triangle-based density setting, which evaluates subgraphs based on their triangle concentration rather than edge count. As a result, the subgraphs returned by {\tt LTDScvx} tend to be smaller in size but possess high triangle densities. This difference explains why {\tt LTDS} may not always identify the same densest subgraphs found by edge-based models, especially in cases where edge density is high but triangle concentration is relatively low.}}

\begin{figure}[t]
    \input{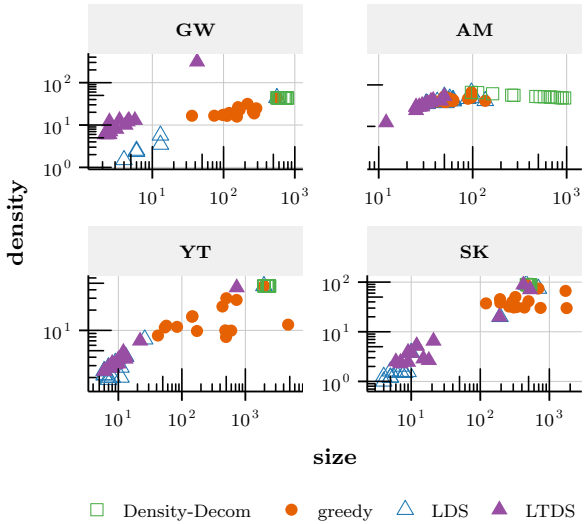}
    \vspace{-1em}
    {\color{blackkkk} \caption{Subgraph statistics: density w.r.t. size.}}
    \label{fig:density-size}
\end{figure}


\subsection{Case Study}
\label{sec:exp:case-study}
Here, we perform a case study on the TL dataset. The TL dataset is provided by a Chinese television company, TCL Technology, which contains three types of vertices: director, movie, and actor.


After examining the top-10 LDSs returned by our LDS algorithm, we found that the LDSs on the TL dataset are about different topics. For example, \Cref{fig:case-study} shows the LDS with the third-highest density. This LDS contains eight movies related to a famous Japanese sci-fiction series ``Ultraman'' with four actors and one director. In this LDS, the four actors participated in all eight films, and the director directed five of them.
The LDS with the second-highest density is a subgraph about Chinese martial fiction. Other LDSs cover topics about western films, Danish comedies, and cartoons, while the DS model can only find a large subgraph about western films.
Hence, from the output of the case study, we reckon that the subgraphs returned by the LDS algorithm are good representations of {\em different} local dense regions in the graph.

\section{Conclusion}
\label{sec:conc}
In this paper, we study the problem of finding the top-$k$ locally densest subgraphs (LDSs) in a graph to identify the local dense regions. The LDSs are usually compact and dense. To facilitate the discovery of LDS, we propose a new concept, named {\em compact number} for each vertex, which denotes the compactness of the most compact subgraph containing the vertex. By leveraging the compact number and its relations with the LDS problem and a specific convex program, we derive a convex-programming based algorithm {\tt LDScvx} following the pruning-and-verifying paradigm. \yy{We also extend our convex-programming based algorithm to the triangle-based setting, named {\tt LTDScvx}, to efficiently discover locally triangle-densest subgraphs (LTDSs).}
\yy{Extensive experiments on thirteen datasets demonstrate that both {\tt LDScvx} and {\tt LTDScvx} consistently outperform existing state-of-the-art algorithms in terms of efficiency.}

{\color{blackkkk} \section{Appendix}
\subsection{Additional Proofs}
\underline{Proof of \Cref{lem:ltds}.} These two properties can be proved by contradiction. For the first property, suppose the tr-density of $G[S]$ is equal to $\rho$, and $G[S']$ is a subgraph of $G[S]$ with a larger tr-density $\rho'$. Then removing $S\setminus S'$ from $G[S]$ would result in the removal of $|T(G[S])|-|T(G[S'])|=\rho |S|-\rho'|S'|<\rho|S\setminus S'|$ triangles, which contradicts \Cref{def:compact}. For the second property, this can be directly derived from \Cref{def:maximal}.

\bigskip
\noindent\underline{Proof of \Cref{th:cnltds}.} We first explore the relationship between the LTDS and the tr-compact numbers of vertices either within or adjacent to the LTDS through the following lemmas.
\begin{lemma}
    \label{lemma:cnltds1}
    Given an LTDS $G[S]$ in $G$, $\forall u \in S$, we have $\phi_{tr}(u)=\mathsf{tr\text{-}density}(G[S])$.
\end{lemma}
\begin{proof}
    We prove the lemma by contradiction.
    Let $\rho$ denote the value of $\mathsf{tr\text{-}density}(G[S])$.
    Since $G[S]$ is $\rho$-tr-compact, $\forall u \in S$, $\phi_{tr}(u) \geq \rho$.
    Suppose there exists a vertex $v \in S$ with $\phi_{tr}(v) > \rho$. Then, we can find a $\phi(v)$-tr-compact subgraph $G[S']$ that contains $v$. Because $G[S]$ is a maximal $\rho$-tr-compact subgraph and $S' \cap S \ne \emptyset$, we have $S' \subseteq S$, which contradicts with \Cref{lem:ltds}.
\end{proof}
\begin{lemma}
    \label{lemma:cnltds2}
    Given an LTDS $G[S]$ in $G$, $\forall (u,v) \in E$, if $u\in S$ and $v\in V\setminus S$, we have $\phi_{tr}(u)>\phi_{tr}(v)$.
\end{lemma}
\begin{proof}
    We prove the lemma by contradiction.
    Let $\rho$ denote the value of $\mathsf{tr\text{-}density}(G[S])$.
    Suppose there exists an edge $ (u, v) \in E $ such that $ u \in S $, $ v \in V \setminus S $, and $ \phi_{tr}(u) \leq \phi_{tr}(v) = \rho$. Then we can find a $\rho$-tr-compact subgraph $G[S']$ of $G$ that contains $v$. $G[S\cup S']$ is a connected supergraph of $g$ and is also $\phi_{tr}(u)$-tr-compact, which contradicts with \Cref{lem:ltds}.
\end{proof}
The sufficiency direction of \Cref{th:cnltds} can be concluded from \Cref{lem:ltds}, while the necessity direction directly follows from \Cref{lemma:cnltds1} and \Cref{lemma:cnltds2}.

\bigskip
\noindent \underline{Proof of \Cref{th:cp}.} We use \Cref{lem:cp} to formally explain the ideal weight distribution of \Cref{equ:trcp}.

\begin{lemma}
    \label{lem:cp}
    Suppose $(\mathbf{r^*}, \bm{\alpha}^*)$ is an optimal solution of \Cref{equ:trcp}.For a vertex $u \in V$, let $X  = \{v \in V | r_v^* > r_u^*\}$, $Y  = \{v \in V | r_v^* = r_u^*\}$, $Z  = \{v \in V | r_v^* < r_u^*\}$. The optimality of $\mathbf{r}^*$ implies that:
    \begin{enumerate}
        \item $\forall (a, b, c) \in T \cap (X \times X \times Y)$, $r_a^* > r_c^*$, $r_b^* > r_c^*$, and $\alpha^*_{c, (a,b,c)} = 1$.
        \item $\forall (a, b, c) \in T \cap (X \times Y \times Y)$, $r_a^* > r_b^* = r_c^*$, and $\alpha^*_{b, (a,b,c)} + \alpha^*_{c, (a,b,c)}= 1$.
        \item $\forall (a, b, c) \in T \cap (Z \times Z \times Y)$, $r_a^* < r_c^*$, $r_b^* < r_c^*$, and $\alpha^*_{a, (a,b,c)} + \alpha^*_{b, (a,b,c)} = 1$.
        \item $\forall (a, b, c) \in T \cap (Z \times Y \times Y)$, $r_a^* < r_b^* = r_c^*$, and $\alpha^*_{a, (a,b,c)} = 1$.
    \end{enumerate}
\end{lemma}

\begin{proof}
    We prove \textit{1} by contradiction: Suppose $\exists t \in T \cap (X \times X \times Y)$, $r_a^* > r_c^*$, $r_b^* > r_c^*$, $\alpha_{c, (a,b,c)} < 1$.
    Since $\alpha_{a, t} + \alpha_{b, t} + \alpha_{c, t} = 1$, we have $\alpha_{a, t} + \alpha_{b, t} > 0$. Obviously there exists some nonnegative $\epsilon$ and $\mu$,  so that $\alpha'_{a, t} = \alpha_{a, t} - \epsilon$, $\alpha'_{b, t} = \alpha_{b, t} - \mu$, and $\alpha'_{c, t} = \alpha_{c, t} + \epsilon +\mu$, satisfies $\alpha'_{a, t} + \alpha'_{b, t} + \alpha'_{c, t} = 1$ and $\alpha'_{a, t}, \alpha'_{b, t}, \alpha'_{c, t} \in [0, 1]$.
    We change $\alpha_{a, t}$, $\alpha_{b, t}$, and $\alpha_{c, t}$ into $\alpha'_{a, t}$, $\alpha'_{b, t}$, and $\alpha'_{c, t}$, respectively. After such modifications, the objective function will change $\epsilon \cdot (-2 r_a^* + 2 r_c^* +  2\epsilon +\mu) + \mu \cdot (-2 r_b^* + 2 r_c^* +  2\mu +\epsilon)$.
    We can always find $\epsilon > 0$ or $\mu > 0$, which satisfies $2 r_a^* - 2 r_c^* >  2\epsilon +\mu$, and $2 r_b^* - 2 r_c^* >  2\mu + \epsilon$, to strictly decrease the objective function. This contradicts that $\mathbf{r^*}$ is the optimal solution to \Cref{equ:trcp}.
    Equation \textit{2}-\textit{4} can be proved similarly.
\end{proof}

\Cref{lem:cp} shows the optimal weight distribution of different kinds of triangles, i.e., each triangle distributes all its weight to vertices with the lowest $r^*$ value.

\begin{lemma}
    \label{lem:cp2}
    Suppose $(\mathbf{r^*}, \bm{\alpha}^*)$ is an optimal solution of $CP(G)$. For a vertex $u\in V$, each connected component of $G\left[X\cup Y\right]$ is $r_u^*$-tr-compact.
\end{lemma}
\begin{proof}
    According to \Cref{lem:cp}, $\left|T\left(G\left[X\cup Y\right]\right)\right|=|T\cap (X\times X \times X)| + |T\cap(X\times X \times Y)| + |T\cap(X\times Y \times Y)| + |T\cap(Y \times Y \times Y)|=\sum_{u\in X\cup Y}r_u^*$, and removing any subset $S\subseteq X\cup Y$ will result in the removal of at least $r_u^*\times |S|$ triangles, this is because\\
    \begin{equation*}
        \begin{aligned}
            \sum_{t\in T(G[X\cup Y])\wedge t\cap S\neq \emptyset}1
             & = \sum_{t\in T(G[X\cup Y])\wedge t\cap S\neq \emptyset}\sum_{v\in t}\alpha_{v,t}                 \\
             & \ge \sum_{t\in T(G[X\cup Y])\wedge t\cap S\neq \emptyset}\sum_{v\in t\wedge v \in S}\alpha_{v,t} \\
             & =\sum_{v\in S}r_v^*                                                                              \\
             & \ge r_u^*\times|S|
        \end{aligned}
    \end{equation*}
    Hence, the lemma holds.
\end{proof}

Based on \Cref{lem:cp}, \Cref{lem:cp2} builds an initial relationship between tr-compactness and \Cref{equ:trcp}. Now, we begin to prove \Cref{th:cp}. \Cref{lem:cp2} guarantees that any $u\in V$ is contained in at least one $r_u^*$-tr-compact subgraph. For any other subgraph $G[S]\nsubseteq G[X\cup Y]$ containing $u$,  $G[S]$ is a $\phi$-tr-compact subgraph, where $\phi \leq r_u^*$. Clearly, $S \cap (Y \cup Z) \neq \emptyset$. Removing $S \cap (Y \cup Z)$ from $G[S]$ will result in the removal of no more than $r_u^* \times |S \cap (Y \cup Z)|$ triangles. Hence, the theorem holds.

\bigskip

Based on \Cref{def:trsg}, we prove $\phi_{tr}(u) \geq \min_{v \in S}r_v$ by contradiction. According to \Cref{th:cp}, $\forall u \in V$, $\phi_{tr}(u) = r_u^*$. Suppose there exists a vertex $u \in S$ such that $r_u^* = \phi_{tr}(u) < \min_{v \in S} r_v \leq r_u$. According to \Cref{def:trsg}, we have $\sum_{v \in S} r_v = \sum_{v \in S} r_v^*$, so there must exist another vertex $x \in S$ such that $r_x^* = \phi(x) > r_x$, which implies $r_x^* \ge \min_{v \in S}r_v > r_u^*$.We can choose $\epsilon$ such that $0 < \epsilon < r_x^* - r_u^*$. If we increase $r_u^*$ by $\epsilon$, and decrease $r_x^*$ by $\epsilon$, the objective value $||\bm{r^*}||^2_2$ will decrease by $2\epsilon (\epsilon + r_u^* - r_x^*)$. This contradicts that $\mathbf{r^*}$ is the optimal solution to \Cref{equ:trcp}, so $\phi(u) \geq \min_{v \in S}r_v$. $\phi(u) \leq \max_{v \in S}r_v$ can be proved in the same way.

\bigskip
\noindent \underline{Proof of \Cref{lem:network}.} For the sake of the proof, we introduce the following definitions and notations. For a given set of vertices $S'$, let $t_i(S')$ be the number of patterns that involve exactly $i$ vertices not from $S'$, $i\in \{1, 2, 3\}$, and $t_u$ be the number of patterns that involve vertex $u$. Let $A=\mathcal{S}\cap S$.

We begin the proof with a structural lemma for the optimal min-cut of the flow network $\mathcal{F}$ in \Cref{alg:is-maximal}.

\begin{lemma}
    Consider any min-cut $(\mathcal{S},\mathcal{T})$ in the network $\mathcal{F}$, the cost of the min-cut is equal to $|P|-t_0(A)+\rho|A|$.
\end{lemma}
\begin{proof}
    For patterns satisfying $p\cap A=p$, node $\mathsf{p}$ corresponding to the specific pattern $p$ has to be in $\mathcal{S}$. If not, then we could reduce the cost of the min-cut by moving $\mathsf{p}$ to $\mathcal{S}$.
    There are six additional cases we considered, one per each type of pattern with respect to set $A$.

    \begin{enumerate}
        \item {$p=\{u, v, w\}\in P$, $u, v, w \notin A$.}
        \item {$p=\{u, v, w\}\in P$, $u, v\notin A$, $w\in A$.}
        \item {$p=\{u, v, w\}\in P$, $u\notin A$, $v, w\in A$.}
        \item {$p=\{u, v\}\in P$, $u, v\notin A$.}
        \item {$p=\{u, v\}\in P$, $u\notin A$, $v\in A$.}
        \item {$p=\{u\}\in P$, $u\notin A$.}
    \end{enumerate}
    For cases 3, 5, 6, node $\mathsf{p}$ corresponding to this specific $p$ can be either in $\mathcal{S}$ or $\mathcal{T}$, and the cost of edges from source $s$ to $\mathcal{T}\cap S$ is equal to $t_1(A)$. For cases 2, 4, node $\mathsf{p}$ corresponding to this specific $p$ has to be in $\mathcal{T}$, and the cost of edges from source $s$ to $\mathcal{T}\cap S$ is equal to $t_2(A)$. For case 1, node $\mathsf{p}$ corresponding to this specific $p$ has to be in $\mathcal{T}$, and the cost of edges from source $s$ to $\mathcal{T}\cap S$ is equal to $t_3(A)$. Furthermore, the cost of edges from $\mathcal{S}$ to sink $t$ equals $\rho|A|$. Summing up these cost terms, the total cost is equal to $t_1(A) + t_2(A) + t_3(A) + \rho|A|=|P|-t_0(A)+\rho|A|$.
\end{proof}

Next, we derive the following lemma, which aims to build a relationship between $S$, $P$, and the original graph.
\begin{lemma}
    \label{lem:PT}
    Removing any subset $S'$ from $S$, such that $\forall u \in S', \phi_{tr}(u)\ge \rho$,  will result in the removal of at least $\rho |S'|$ patterns from $P$.
\end{lemma}
\begin{proof}
    According to \Cref{def:trcn}, $\phi_{tr}(u)\ge \rho$ suggests that $u$ is contained in at least one $\rho$-tr-compact subgraph of $G$. For all triangles containing $u$ and possibly contained in a $\rho$-tr-compact subgraph, we have generated a corresponding $p\in P$ in \Cref{alg:is-ltds}. Suppose the lemma holds to the contrary; this would contradict the definition of \Cref{def:compact}.
\end{proof}

\begin{lemma}
    \label{lem:optimal}
    Given an optimal min-cut $(\mathcal{S}, \mathcal{T})$ in the network $\mathcal{F}$ that maximize $|A|$, vertex $u\in A$ if and only if $\phi_{tr}(u)\ge \rho$.
\end{lemma}
\begin{proof}
    Let $B$ denote the vertex set of all vertices $u$ such that $u\in S\cap \phi_{tr}(u)\ge \rho$. We prove that $A$ and $B$ are the same.

    \textit{(1) We prove that $A\subseteq B$.} Suppose to the contrary that some vertices $u\in A$ have a tr-compact number smaller than $\rho$, then according to \Cref{lem:PT},there exists a subset $S'\subseteq A$ such that removing $S'$ from $A$ would result in the removal of less that $\rho|S'|$ patterns. Thus, we have $(|P|-t_0(A\setminus S')+\rho|A\setminus S'|)-(|P|-t_0(A)+\rho|A|)=t_0(A)-t_0(A\setminus S') - \rho |S'|<0$, which contradicts that $(\mathcal{S}, \mathcal{T})$ is optimal.

    \textit{(2) We prove that $B\subseteq A$.} Suppose to the contrary that $B$ is not a subset of $A$. According to \textit{(1)}, we can derive that $A\subsetneq B$. According to \Cref{lem:PT}, removing $B\setminus A$ from $B$ would result in the removal of at least $\rho|B\setminus A|$ patterns. Thus, we have $(|P|-t_0(B)+\rho|B|)-(|P|-t_0(A)+\rho|A|)=t_0(A)-t_0(B) + \rho |B\setminus A|<=0$. To maximize $|A|$, $B\setminus A$ has to be $\emptyset$, which contradicts our assumption.

    According to \textit{(1)} and \textit{(2)}, the lemma is proved.
\end{proof}

\Cref{lem:network} follows directly from \Cref{lem:optimal}.

\bigskip
\noindent \underline{Proof of \Cref{th:isltds}.} First, if $G[S]$ is an LTDS, \Cref{alg:is-ltds} returns {\sf True}. Because according to \Cref{lem:network}, \Cref{alg:is-maximal} returns all vertices $u\in U$ such that their tr-compact number is larger than or equal to $\rho$. $G[S]$ is a connected component in $G[S']$ suggests that $G[S]$ is a maximal $\rho$-tr-compact subgraph in $G$. Otherwise the maximal $\mathsf{density}(G[S])$-tr-compact subgraph containing $G[S]$ must contain a vertex $u$ with corresponding patterns, and then we can construct a larger $\mathsf{density}(G[S])$-tr-compact subgraph in $G$ by including vertices with $\underline{\phi}_{tr}(w) > \mathsf{density}(G[S])$ connected to $u$. Hence, the contradiction proves the claim.

On the other direction, if $G[S]$ is not an LTDS of $G$, we will find a larger $\mathsf{density}(G[S])$-tr-compact subgraph containing $G[S]$. Thus, the algorithm returns {\sf False}.}
\bibliographystyle{spbasic}      
\bibliography{ref}   

\begin{thebibliography}{62}
\providecommand{\natexlab}[1]{#1}
\providecommand{\url}[1]{{#1}}
\providecommand{\urlprefix}{URL }
\expandafter\ifx\csname urlstyle\endcsname\relax
  \providecommand{\doi}[1]{DOI~\discretionary{}{}{}#1}\else
  \providecommand{\doi}{DOI~\discretionary{}{}{}\begingroup
  \urlstyle{rm}\Url}\fi
\providecommand{\eprint}[2][]{\url{#2}}

\bibitem[{Albert and Barab{\'a}si(2002)}]{albert2002statistical}
Albert R, Barab{\'a}si AL (2002) Statistical mechanics of complex networks.
  Reviews of modern physics 74(1):47

\bibitem[{Andersen and Chellapilla(2009)}]{andersen2009finding}
Andersen R, Chellapilla K (2009) Finding dense subgraphs with size bounds. In:
  International workshop on algorithms and models for the web-graph, Springer,
  pp 25--37

\bibitem[{Asahiro et~al(2000)Asahiro, Iwama, Tamaki, and
  Tokuyama}]{asahiro2000greedily}
Asahiro Y, Iwama K, Tamaki H, Tokuyama T (2000) Greedily finding a dense
  subgraph. Journal of Algorithms 34(2):203--221

\bibitem[{Asahiro et~al(2002)Asahiro, Hassin, and
  Iwama}]{asahiro2002complexity}
Asahiro Y, Hassin R, Iwama K (2002) Complexity of finding dense subgraphs.
  Discrete Applied Mathematics 121(1-3):15--26

\bibitem[{Bahmani et~al(2012)Bahmani, Kumar, and
  Vassilvitskii}]{bahmani2012densest}
Bahmani B, Kumar R, Vassilvitskii S (2012) Densest subgraph in streaming and
  mapreduce. arXiv preprint arXiv:12016567

\bibitem[{Boldi and Vigna(2004)}]{boldi2004webgraph}
Boldi P, Vigna S (2004) The webgraph framework i: Compression techniques. In:
  Proceedings of the 13th international conference on World Wide Web, ACM, pp
  595--602

\bibitem[{Boldi et~al(2011)Boldi, Rosa, Santini, and Vigna}]{boldi2011layered}
Boldi P, Rosa M, Santini M, Vigna S (2011) Layered label propagation: A
  multiresolution coordinate-free ordering for compressing social networks. In:
  Proceedings of the 20th international conference on World Wide Web, ACM, pp
  587--596

\bibitem[{Boob et~al(2020)Boob, Gao, Peng, Sawlani, Tsourakakis, Wang, and
  Wang}]{boob2020flowless}
Boob D, Gao Y, Peng R, Sawlani S, Tsourakakis C, Wang D, Wang J (2020)
  Flowless: Extracting densest subgraphs without flow computations. In:
  Proceedings of The Web Conference 2020, pp 573--583

\bibitem[{Chang and Qiao(2020)}]{chang2020deconstruct}
Chang L, Qiao M (2020) Deconstruct densest subgraphs. In: Proceedings of The
  Web Conference 2020, pp 2747--2753

\bibitem[{Charikar(2000)}]{charikar2000greedy}
Charikar M (2000) Greedy approximation algorithms for finding dense components
  in a graph. In: International Workshop on Approximation Algorithms for
  Combinatorial Optimization, Springer, pp 84--95

\bibitem[{Chekuri et~al(2022)Chekuri, Quanrud, and Torres}]{chekuri2022densest}
Chekuri C, Quanrud K, Torres MR (2022) Densest subgraph: Supermodularity,
  iterative peeling, and flow. In: SODA, SIAM, pp 1531--1555

\bibitem[{Chen and Saad(2010)}]{chen2010dense}
Chen J, Saad Y (2010) Dense subgraph extraction with application to community
  detection. IEEE Transactions on knowledge and data engineering
  24(7):1216--1230

\bibitem[{Ching et~al(2015)Ching, Edunov, Kabiljo, Logothetis, and
  Muthukrishnan}]{ching2015one}
Ching A, Edunov S, Kabiljo M, Logothetis D, Muthukrishnan S (2015) One trillion
  edges: Graph processing at facebook-scale. Proceedings of the VLDB Endowment
  8(12):1804--1815

\bibitem[{Conte et~al(2018)Conte, De~Matteis, De~Sensi, Grossi, Marino, and
  Versari}]{conte2018d2k}
Conte A, De~Matteis T, De~Sensi D, Grossi R, Marino A, Versari L (2018) D2k:
  scalable community detection in massive networks via small-diameter k-plexes.
  In: Proceedings of the 24th ACM SIGKDD International Conference on Knowledge
  Discovery \& Data Mining, pp 1272--1281

\bibitem[{Danisch et~al(2017)Danisch, Chan, and Sozio}]{danisch2017large}
Danisch M, Chan THH, Sozio M (2017) Large scale density-friendly graph
  decomposition via convex programming. In: Proceedings of the 26th
  International Conference on World Wide Web, pp 233--242

\bibitem[{Dondi et~al(2021{\natexlab{a}})Dondi, Hosseinzadeh, and
  Guzzi}]{dondi2021novel}
Dondi R, Hosseinzadeh MM, Guzzi PH (2021{\natexlab{a}}) A novel algorithm for
  finding top-k weighted overlapping densest connected subgraphs in dual
  networks. Applied Network Science 6(1):1--17

\bibitem[{Dondi et~al(2021{\natexlab{b}})Dondi, Hosseinzadeh, Mauri, and
  Zoppis}]{dondi2021top}
Dondi R, Hosseinzadeh MM, Mauri G, Zoppis I (2021{\natexlab{b}}) Top-k
  overlapping densest subgraphs: approximation algorithms and computational
  complexity. Journal of Combinatorial Optimization 41(1):80--104

\bibitem[{Dourisboure et~al(2007)Dourisboure, Geraci, and
  Pellegrini}]{dourisboure2007extraction}
Dourisboure Y, Geraci F, Pellegrini M (2007) Extraction and classification of
  dense communities in the web. In: Proceedings of the 16th international
  conference on World Wide Web, pp 461--470

\bibitem[{Fang et~al(2019{\natexlab{a}})Fang, Yu, Cheng, Lakshmanan, and
  Lin}]{fang12efficient}
Fang Y, Yu K, Cheng R, Lakshmanan LV, Lin X (2019{\natexlab{a}}) Efficient
  algorithms for densest subgraph discovery. Proceedings of the VLDB Endowment
  12(11)

\bibitem[{Fang et~al(2019{\natexlab{b}})Fang, Yu, Cheng, Lakshmanan, and
  Lin}]{fang2019efficient}
Fang Y, Yu K, Cheng R, Lakshmanan LV, Lin X (2019{\natexlab{b}}) Efficient
  algorithms for densest subgraph discovery. arXiv preprint arXiv:190600341

\bibitem[{Fang et~al(2020)Fang, Huang, Qin, Zhang, Zhang, Cheng, and
  Lin}]{fang2020survey}
Fang Y, Huang X, Qin L, Zhang Y, Zhang W, Cheng R, Lin X (2020) A survey of
  community search over big graphs. The VLDB Journal 29(1):353--392

\bibitem[{Fang et~al(2022)Fang, Luo, and Ma}]{fang2022densest}
Fang Y, Luo W, Ma C (2022) Densest subgraph discovery on large graphs:
  Applications, challenges, and techniques. Proceedings of the VLDB Endowment
  15

\bibitem[{Frank et~al(1956)Frank, Wolfe et~al}]{frank1956algorithm}
Frank M, Wolfe P, et~al (1956) An algorithm for quadratic programming. Naval
  research logistics quarterly 3(1-2):95--110

\bibitem[{Fratkin et~al(2006)Fratkin, Naughton, Brutlag, and
  Batzoglou}]{fratkin2006motifcut}
Fratkin E, Naughton BT, Brutlag DL, Batzoglou S (2006) Motifcut: regulatory
  motifs finding with maximum density subgraphs. Bioinformatics
  22(14):e150--e157

\bibitem[{Galbrun et~al(2016)Galbrun, Gionis, and Tatti}]{galbrun2016top}
Galbrun E, Gionis A, Tatti N (2016) Top-k overlapping densest subgraphs. Data
  Mining and Knowledge Discovery 30(5):1134--1165

\bibitem[{Gionis et~al(2013)Gionis, Junqueira, Leroy, Serafini, and
  Weber}]{gionis2013piggybacking}
Gionis A, Junqueira FP, Leroy V, Serafini M, Weber I (2013) Piggybacking on
  social networks. In: VLDB 2013-39th International Conference on Very Large
  Databases, vol~6, pp 409--420

\bibitem[{Goldberg(1984)}]{goldberg1984finding}
Goldberg AV (1984) Finding a maximum density subgraph. University of California
  Berkeley

\bibitem[{Han(2019)}]{han2019traffic}
Han X (2019) Traffic incident detection: a deep learning framework. In: 2019
  20th IEEE International Conference on Mobile Data Management (MDM), IEEE, pp
  379--380

\bibitem[{Han et~al(2020)Han, Grubenmann, Cheng, Wong, Li, and
  Sun}]{han2020traffic}
Han X, Grubenmann T, Cheng R, Wong SC, Li X, Sun W (2020) Traffic incident
  detection: A trajectory-based approach. In: 2020 IEEE 36th International
  Conference on Data Engineering (ICDE), IEEE, pp 1866--1869

\bibitem[{Han et~al(2022{\natexlab{a}})Han, Cheng, Grubenmann, Maniu, Ma, and
  Li}]{han2022leveraging}
Han X, Cheng R, Grubenmann T, Maniu S, Ma C, Li X (2022{\natexlab{a}})
  Leveraging contextual graphs for stochastic weight completion in sparse road
  networks. In: Proceedings of the 2022 SIAM International Conference on Data
  Mining (SDM), SIAM, pp 64--72

\bibitem[{Han et~al(2022{\natexlab{b}})Han, Cheng, Ma, and
  Grubenmann}]{han2022deeptea}
Han X, Cheng R, Ma C, Grubenmann T (2022{\natexlab{b}}) Deeptea: effective and
  efficient online time-dependent trajectory outlier detection. Proceedings of
  the VLDB Endowment 15(7):1493--1505

\bibitem[{Han et~al(2022{\natexlab{c}})Han, Dell’Aglio, Grubenmann, Cheng,
  and Bernstein}]{han2022framework}
Han X, Dell’Aglio D, Grubenmann T, Cheng R, Bernstein A (2022{\natexlab{c}})
  A framework for differentially-private knowledge graph embeddings. Journal of
  Web Semantics 72:100,696

\bibitem[{Hooi et~al(2016)Hooi, Song, Beutel, Shah, Shin, and
  Faloutsos}]{hooi2016fraudar}
Hooi B, Song HA, Beutel A, Shah N, Shin K, Faloutsos C (2016) Fraudar: Bounding
  graph fraud in the face of camouflage. In: Proceedings of the 22nd ACM SIGKDD
  international conference on knowledge discovery and data mining, pp 895--904

\bibitem[{Jin et~al(2009)Jin, Xiang, Ruan, and Fuhry}]{jin20093}
Jin R, Xiang Y, Ruan N, Fuhry D (2009) 3-hop: a high-compression indexing
  scheme for reachability query. In: Proceedings of the 2009 ACM SIGMOD
  International Conference on Management of data, pp 813--826

\bibitem[{Kannan and Vinay(1999)}]{kannan1999analyzing}
Kannan R, Vinay V (1999) Analyzing the structure of large graphs.
  Forschungsinst. f{\"u}r Diskrete Mathematik

\bibitem[{Khuller and Saha(2009)}]{khuller2009finding}
Khuller S, Saha B (2009) On finding dense subgraphs. In: International
  colloquium on automata, languages, and programming, Springer, pp 597--608

\bibitem[{Leskovec and Krevl(2014)}]{snapnets}
Leskovec J, Krevl A (2014) {SNAP Datasets}: {Stanford} large network dataset
  collection. \url{http://snap.stanford.edu/data}

\bibitem[{Li et~al(2021)Li, Cheng, Chang, Shan, Ma, and Cao}]{li2021analyzing}
Li X, Cheng R, Chang KCC, Shan C, Ma C, Cao H (2021) On analyzing graphs with
  motif-paths. Proceedings of the VLDB Endowment 14(6):1111--1123

\bibitem[{Ma et~al(2019)Ma, Cheng, Lakshmanan, Grubenmann, Fang, and
  Li}]{ma2019linc}
Ma C, Cheng R, Lakshmanan LV, Grubenmann T, Fang Y, Li X (2019) Linc: a motif
  counting algorithm for uncertain graphs. Proceedings of the VLDB Endowment
  13(2):155--168

\bibitem[{Ma et~al(2020)Ma, Fang, Cheng, Lakshmanan, Zhang, and
  Lin}]{ma2020efficient}
Ma C, Fang Y, Cheng R, Lakshmanan LV, Zhang W, Lin X (2020) Efficient
  algorithms for densest subgraph discovery on large directed graphs. In:
  Proceedings of the 2020 ACM SIGMOD International Conference on Management of
  Data, pp 1051--1066

\bibitem[{Ma et~al(2021{\natexlab{a}})Ma, Fang, Cheng, Lakshmanan, Zhang, and
  Lin}]{ma2021efficient}
Ma C, Fang Y, Cheng R, Lakshmanan LV, Zhang W, Lin X (2021{\natexlab{a}})
  Efficient directed densest subgraph discovery. ACM SIGMOD Record 50(1):33--40

\bibitem[{Ma et~al(2021{\natexlab{b}})Ma, Fang, Cheng, Lakshmanan, Zhang, and
  Lin}]{ma2021directed}
Ma C, Fang Y, Cheng R, Lakshmanan LV, Zhang W, Lin X (2021{\natexlab{b}}) On
  directed densest subgraph discovery. ACM Transactions on Database Systems
  (TODS) 46(4):1--45

\bibitem[{Ma et~al(2022)Ma, Fang, Cheng, Lakshmanan, and Han}]{ma2022convex}
Ma C, Fang Y, Cheng R, Lakshmanan LV, Han X (2022) A convex-programming
  approach for efficient directed densest subgraph discovery. In: Proceedings
  of the 2022 International Conference on Management of Data, pp 845--859

\bibitem[{Mitzenmacher et~al(2015)Mitzenmacher, Pachocki, Peng, Tsourakakis,
  and Xu}]{mitzenmacher2015scalable}
Mitzenmacher M, Pachocki J, Peng R, Tsourakakis C, Xu SC (2015) Scalable large
  near-clique detection in large-scale networks via sampling. In: KDD, pp
  815--824

\bibitem[{Orlin(2013)}]{orlin2013max}
Orlin JB (2013) Max flows in o (nm) time, or better. In: Proceedings of the
  forty-fifth annual ACM symposium on Theory of computing, pp 765--774

\bibitem[{Qin et~al(2015)Qin, Li, Chang, and Zhang}]{qin2015locally}
Qin L, Li RH, Chang L, Zhang C (2015) Locally densest subgraph discovery. In:
  Proceedings of the 21th ACM SIGKDD International Conference on Knowledge
  Discovery and Data Mining, pp 965--974

\bibitem[{Rossi and Ahmed(2015)}]{nr}
Rossi RA, Ahmed NK (2015) The network data repository with interactive graph
  analytics and visualization. In: AAAI,
  \urlprefix\url{https://networkrepository.com}

\bibitem[{Saha et~al(2010)Saha, Hoch, Khuller, Raschid, and
  Zhang}]{saha2010dense}
Saha B, Hoch A, Khuller S, Raschid L, Zhang XN (2010) Dense subgraphs with
  restrictions and applications to gene annotation graphs. In: Annual
  International Conference on Research in Computational Molecular Biology,
  Springer, pp 456--472

\bibitem[{Samusevich et~al(2016)Samusevich, Danisch, and
  Sozio}]{samusevich2016local}
Samusevich R, Danisch M, Sozio M (2016) Local triangle-densest subgraphs. In:
  2016 IEEE/ACM International Conference on Advances in Social Networks
  Analysis and Mining (ASONAM), IEEE, pp 33--40

\bibitem[{Sawlani and Wang(2020)}]{sawlani2020near}
Sawlani S, Wang J (2020) Near-optimal fully dynamic densest subgraph. In:
  Proceedings of the 52nd Annual ACM SIGACT Symposium on Theory of Computing,
  pp 181--193

\bibitem[{Seidman(1983)}]{seidman1983network}
Seidman SB (1983) Network structure and minimum degree. Social networks
  5(3):269--287

\bibitem[{Stelzl et~al(2005)Stelzl, Worm, Lalowski, Haenig, Brembeck, Goehler,
  Stroedicke, Zenkner, Schoenherr, Koeppen et~al}]{stelzl2005human}
Stelzl U, Worm U, Lalowski M, Haenig C, Brembeck FH, Goehler H, Stroedicke M,
  Zenkner M, Schoenherr A, Koeppen S, et~al (2005) A human protein-protein
  interaction network: a resource for annotating the proteome. Cell
  122(6):957--968

\bibitem[{Sun et~al(2020)Sun, Danisch, Chan, and Sozio}]{sun2020kclist++}
Sun B, Danisch M, Chan TH, Sozio M (2020) Kclist++: A simple algorithm for
  finding k-clique densest subgraphs in large graphs. Proceedings of the VLDB
  Endowment (PVLDB)

\bibitem[{Tatti and Gionis(2015)}]{tatti2015density}
Tatti N, Gionis A (2015) Density-friendly graph decomposition. In: Proceedings
  of the 24th International Conference on World Wide Web, pp 1089--1099

\bibitem[{Tsourakakis(2015)}]{tsourakakis2015k}
Tsourakakis C (2015) The k-clique densest subgraph problem. In: Proceedings of
  the 24th international conference on world wide web, pp 1122--1132

\bibitem[{Tsourakakis and Chen(2021)}]{tsourakakis2021dense}
Tsourakakis C, Chen T (2021) Dense subgraph discovery: Theory and application.
  In: SDM Tutorial

\bibitem[{Tsourakakis et~al(2013)Tsourakakis, Bonchi, Gionis, Gullo, and
  Tsiarli}]{tsourakakis2013denser}
Tsourakakis C, Bonchi F, Gionis A, Gullo F, Tsiarli M (2013) Denser than the
  densest subgraph: extracting optimal quasi-cliques with quality guarantees.
  In: Proceedings of the 19th ACM SIGKDD international conference on Knowledge
  discovery and data mining, pp 104--112

\bibitem[{Tsourakakis(2014)}]{tsourakakis2014novel}
Tsourakakis CE (2014) A novel approach to finding near-cliques: The
  triangle-densest subgraph problem. arXiv preprint arXiv:14051477

\bibitem[{Wang and Cheng(2012)}]{wang2012truss}
Wang J, Cheng J (2012) Truss decomposition in massive networks. Proceedings of
  the VLDB Endowment 5(9)

\bibitem[{Yang and Leskovec(2015)}]{yang2015defining}
Yang J, Leskovec J (2015) Defining and evaluating network communities based on
  ground-truth. Knowledge and Information Systems 42(1):181--213

\bibitem[{Zhou et~al(2024{\natexlab{a}})Zhou, Guo, Fang, and
  Ma}]{zhou2024counting}
Zhou Y, Guo Q, Fang Y, Ma C (2024{\natexlab{a}}) A counting-based approach for
  efficient k-clique densest subgraph discovery. Proceedings of the ACM on
  Management of Data 2(3):1--27

\bibitem[{Zhou et~al(2024{\natexlab{b}})Zhou, Guo, Yang, Fang, Ma, and
  Lakshmanan}]{zhou2024depth}
Zhou Y, Guo Q, Yang Y, Fang Y, Ma C, Lakshmanan L (2024{\natexlab{b}}) In-depth
  analysis of densest subgraph discovery in a unified framework. arXiv preprint
  arXiv:240604738

\end{thebibliography}


\end{document}